\documentclass[oribibl]{llncs}

\usepackage{makeidx}  
\usepackage{amssymb}
\usepackage{graphicx, subfigure}
\usepackage{wrapfig}
\usepackage[ruled,vlined,linesnumbered]{algorithm2e}
\usepackage{float} 
\usepackage{enumitem}
\usepackage[left=2cm,top=3cm,right=2cm]{geometry} 

\SetCommentSty{mycommfont}

\graphicspath{{./graphics/}}

\newtheorem{mydefP}{T{\'o}th's Proposition}
\newtheorem{mydefL}{T{\'o}th's Lemma}
\newtheorem{mydefC}{T{\'o}th's Claim}

\usepackage{hyperref}
\hypersetup{
colorlinks=true, 
linkcolor=blue, 
citecolor=green, 
urlcolor=webbrown, 
bookmarksopen=true,
pdftitle={Trilateration Research}, 
breaklinks=true, bookmarks=true,bookmarksnumbered,
pdfauthor={Alina Shaikhet}, 
pdfsubject={}, 
pdfkeywords={}, 
pdfcreator={pdfLaTeX}, 
pdfproducer={LaTeX with hyperref and ClassicThesis} 
}
\usepackage{bookmark}

\begin{document}

\pagestyle{headings}  

\title{Art Gallery Localization\thanks{This work was supported by NSERC.}}
\titlerunning{Art Gallery Localization}  
%
\author{Prosenjit Bose\inst{1} \and Jean-Lou De Carufel\inst{2} \and Alina Shaikhet\inst{1} \and Michiel Smid\inst{1}}
\authorrunning{P. Bose, J.-L. De Carufel, A. Shaikhet and M. Smid} 

\institute{School of Computer Science, Carleton University, Ottawa, Canada,\\
\email{\{jit, michiel\}@scs.carleton.ca, alina.shaikhet@carleton.ca},\\
\and
School of Electrical Engineering and Computer Science, U. of Ottawa, Canada,\\
\email{jdecaruf@uottawa.ca}}

\maketitle              

\begin{abstract}
We study the problem of placing a set $T$ of broadcast towers in a simple polygon $P$ in order for any point to locate itself in the interior of $P$. Let $V(p)$ denote the \emph{visibility polygon} of a point $p$, as the set of all points $q \in P$ that are visible to $p$. For any point $p \in P$: for each tower $t \in T \cap V(p)$ the point $p$ receives the coordinates of $t$ and the exact Euclidean distance between $t$ and $p$. From this information $p$ can determine its coordinates. We show a tower-positioning algorithm that computes such a set $T$ of size at most $\lfloor 2n/3\rfloor$, where $n$ is the size of $P$. This improves the previous upper bound of $\lfloor 8n/9\rfloor$ towers~\cite{DBLP:conf/cccg/DippelS15}. We also show that $\lfloor 2n/3\rfloor$ towers are sometimes necessary.
\end{abstract}

\section{Introduction}
\label{sec:introduction}
\pdfbookmark[1]{Introduction}{sec:introduction}

The art gallery problem was introduced in 1973 when
Victor Klee asked how many guards are sufficient to \emph{guard} the interior of a simple polygon having $n$ vertices. Although it has been shown by Chv{\'a}tal that $\lfloor n/3\rfloor$ guards are always sufficient and sometimes necessary~\cite{Chvatal197539}, and such a set of guards can be computed easily~\cite{Fisk1978374}, such solutions are usually far from optimal in terms of minimizing the number of guards for a particular input polygon. Moreover, it was shown that determining an optimal number of guards is NP-hard, even for simple polygons~\cite{Lee:1986:CCA:13643.13657}. Determining the actual locations of those guards is even harder~\cite{DBLP:conf/compgeom/AbrahamsenAM17}.

\emph{Trilateration} is the process of determining absolute or relative locations of points by measurement of distances, using the geometry of the environment. In addition to its interest as a geometric problem, trilateration has practical applications in surveying and navigation, including global positioning systems (GPS). Every GPS satellite transmits information about its position and the current time at regular intervals. These signals are intercepted by a GPS receiver, which calculates how far away each satellite is based on how long it took for the messages to arrive. GPS receivers take this information and use trilateration to calculate the user's location.

In our research we combine the art gallery problem with trilateration. We address the problem of placing broadcast \emph{towers} in a simple polygon $P$ in order for a point in $P$ (let us call it an \emph{agent}) to locate itself. Towers can be defined as points, which can transmit their coordinates together with a time stamp to other points in their visibility region. The agent receives messages from all the towers that belong to its visibility region. Given a message from the tower $t$, the agent can determine its distance to $t$. In our context, \emph{trilateration} is the process, during which the agent can determine its absolute coordinates from the messages the agent receives. Receiving a message from one tower only will not be sufficient for the agent to locate itself (unless the agent and the tower are at the the same location). In Euclidean plane, two distinct circles intersect in at most two points. If a point lies on two circles, then the circle centers and the two radii provide sufficient information to narrow the possible locations down to two. Additional information may narrow the possibilities down to one unique location.

In relation to GPS systems, towers can be viewed as GPS satellites, while agents (query points interior to the polygon) can be compared to GPS receivers. Naturally, we would like to minimize the number of towers. 

Let $P$ be a simple polygon in general position (no three vertices are collinear) having a total of $n$ vertices on its boundary (denoted by $\partial P$). Note that $\partial P \subset P$. Two points $u, v \in P$ are \emph{visible to each other} if the segment $\overline{uv}$ is contained in $P$. We also say that $u$ \emph{sees} $v$. Note that $\overline{uv}$ may touch $\partial P$ in one or more points. For $u \in P$, we let $V(u)$ denote the \emph{visibility polygon} of $u$, as the set of all points $q \in P$ that are visible to $u$. Notice that $V(u)$ is a star-shaped polygon contained in $P$ and $u$ belongs to its \emph{kernel} (the set of points from which all of $V(u)$ is visible).

\vspace{10pt}
\textbf{Problem Definition:} Let $T$ be a set of points (called \emph{towers}) in $P$ satisfying the following properties. 
For any point $p \in P$: for each $t \in T \cap V(p)$, the point $p$ receives the coordinates of $t$ and can compute the Euclidean distance between $t$ and $p$, denoted $d(t,p)$. From this information, $p$ can determine its coordinates. We consider the following problems:
\begin{enumerate}
\item Design an algorithm that, on any input polygon $P$ in general position, computes a ``small'' set $T$ of towers.
\item Design a localization algorithm. 
\end{enumerate}

We show how to compute such a set $T$ of size at most $\lfloor 2n/3\rfloor$ by using the polygon partition method introduced by T{\'o}th~\cite{Toth2000121}. T{\'o}th showed that any simple polygon with $n$ vertices can be guarded by $\lfloor n/3\rfloor$ point guards whose range of vision is $180^\circ$. T{\'o}th partitions a polygon into subpolygons, on which he then can apply induction. He cuts along diagonals whenever it is possible, otherwise he cuts along a continuation of some edge of $P$; along a two-line segment made of an extension of two edges of $P$ that intersect inside $P$; or along the bisector of a reflex vertex of $P$. Notice that the three latter types of cuts may introduce new vertices that are not necessarily in general position with the given set of vertices. Succeeding partitions of the subpolygons may create polygons that are not simple. However, T{\'o}th assumed that his partition method creates subpolygons whose vertices are in general position (refer to~\cite{Toth2000121} Section $2$). We lift this assumption and show how to adapt his method to a wider range of polygons, which we define in Section~\ref{sec:partition}. Under the new conditions the partition of $P$ may contain star-shaped polygons whose kernel is a single point. It does not pose an obstacle to T{\'o}th's problem, but it is a severe complication to our problem, because we require a pair of distinct towers in the kernel of each polygon of the partition. We modify T{\'o}th's partition method and show how to use it with respect to our problem. It is important to notice that we assume that the input polygon is in general position, while non-general position may occur for subpolygons of the partition (refer to Definition~\ref{def:polygon}). We show that after the modification each $180^\circ$-guard $g$ can be replaced with a pair of towers \textbf{close} to $g$. We embed the orientation of the $180^\circ$-guard into the coordinates of the towers. That is, we specify to which side of the line $L$ through the pair of towers their primary localization region resides. We do it by positioning the towers at a distance that is an exact rational number. The parity of the numerator of the rational number (in the reduced form) defines which one of the two half-planes (defined by $L$) the pair of towers are responsible for. We call it the \emph{parity trick}. For example, if we want a pair $t_1$, $t_2$ of towers to be responsible for the half-plane to the left of the line through $t_1$ and $t_2$, then we position the towers at a distance, which is a reduced rational number whose numerator is even. The localization algorithm is allowed to use this information.

Our interest in this problem started with the paper by Dippel and Sundaram~\cite{DBLP:conf/cccg/DippelS15}. 
 They provide the first non-trivial bounds on agent localization in simple polygons, by showing that $\lfloor 8n/9\rfloor$ towers suffice for any non-degenerate polygon of $n$ vertices, and present an $O(n \log n)$ algorithm for the corresponding placement. Their approach is to decompose the polygon into at most $\lfloor n/3\rfloor$ fans. A polygon $P'$ is a \emph{fan} if there exist a vertex $u$, such that for every other vertex $v$ not adjacent to $u$, $\overline{u v}$ is a  diagonal of $P'$; the vertex $u$ is called the \emph{center} of the fan. In each fan with fewer than $4$ triangles Dippel and Sundaram position a pair of towers on an edge of the fan; every fan with $4$ or more triangles receives a triple of towers in its kernel. In a classical trilateration, the algorithm for locating an agent knows the coordinates of the towers that can see $p$ together with distances between $p$ and the corresponding towers. However, the localization algorithm presented in~\cite{DBLP:conf/cccg/DippelS15} requires a lot of additional information, such as a complete information about the polygon, its decomposition into fans and the coordinates of \textbf{all} towers.

Our localization algorithm has no information about $P$. It receives as input only the coordinates of the towers that can see $p$ together with their distances to $p$. In addition our algorithm is empowered by the knowledge of the parity trick. 
When only a pair $t_1$, $t_2$ of towers can see $p$ then the coordinates of the towers together with the distances $d(t_1,p)$ and $d(t_2,p)$ provide sufficient information to narrow the possible locations of $p$ down to two. Refer to Figures~\ref{fig:example2},\ref{fig:example3}. Those two locations are reflections of each other over the line through $t_1$ and $t_2$. In this situation our localization algorithm uses the parity trick. It calculates the distance between the two towers and judging by the parity of this number decides which of the two possible locations is the correct position of $p$.

We show how to position at most $\lfloor 2n/3\rfloor$ towers inside $P$, which is an improvement over the previous upper bound in~\cite{DBLP:conf/cccg/DippelS15}. We also show that $\lfloor 2n/3\rfloor$ towers are sometimes necessary. The comb polygon from the original art gallery problem can be used to show a lower bound. Refer to Fig.~\ref{fig:comb}. No point in the comb can see two different comb spikes. Thus we need at least two towers per spike to localize all of the points in its interior. In addition we need to know the parity trick. Or, alternatively, we need to know $P$, its exact location and orientation. We show in Theorem~\ref{thm:no_map} that without any additional information (such as the parity trick or the complete knowledge about $P$ including its partition) it is \textbf{not} possible to localize an agent in a simple $n$-gon (where $n = 3k + q$, for integer $k \geq 1$ and $q = 0$, $1$ or $2$) with less than $n - q$ towers. 

In Section~\ref{sec:preliminaries} we give basic definitions and present some properties and observations. Section~\ref{sec:partition} shows some of our modifications of the polygon partition given by T{\'o}th~\cite{Toth2000121} and its adaptation to our problem. In Section~\ref{sec:localization} we present a localization algorithm.

\section{Preliminaries}
\label{sec:preliminaries}
\pdfbookmark[1]{Preliminaries}{sec:preliminaries}

Consider a point (an agent) $p$ in the interior of $P$, whose location is unknown. Let $C(x,r)$ denote the circle centered at a point $x$ with radius $r$. If only one tower $t$ can see $p$ then $p$ can be anywhere on $C(t,d(p,t)) \cap V(t)$ (refer to Fig.~\ref{fig:example1}), which may not be enough to identify the location of~$p$. 
By the \emph{map} of $P$ we mean the complete information about $P$ including the coordinates of all the vertices of $P$ and the vertex adjacency list. Notice that one must know the map of $P$ to calculate $V(t)$. 
If two towers $t_1$ and $t_2$ can see $p$ then the location of $p$ can be narrowed down to at most two points $C(t_1,d(p,t_1)) \cap C(t_2,d(p,t_2)) \cap V(t_1) \cap V(t_2)$. Refer to Fig.~\ref{fig:example2}. The two points are reflections of each other over the line through $t_1$ and $t_2$. We call this situation the \emph{ambiguity along the line}, because without any additional information we do not know to which one of the two locations $p$ belongs. To avoid this uncertainty we can place both towers on the same edge of $P$. Consider for example Fig.~\ref{fig:example3} where two towers are placed on the line segment $kernel(P) \cap \partial P$. In this example, if the map of $P$ is known (and thus we know $V(t_1)$ and $V(t_2)$) then the intersection $C(t_1,d(p,t_1)) \cap C(t_2,d(p,t_2)) \cap V(t_1) \cap V(t_2)$ is a single point (highlighted in red). Alternatively, if the map of $P$ is unknown, we can place a triple of non-collinear towers in the kernel of $P$ (highlighted in cyan on Fig.~\ref{fig:example3}) to localize any point interior to $P$.

\begin{figure}[h]
\centering
\subfigure[]{%
		\label{fig:example1}%
		\includegraphics[width=0.22\textwidth]{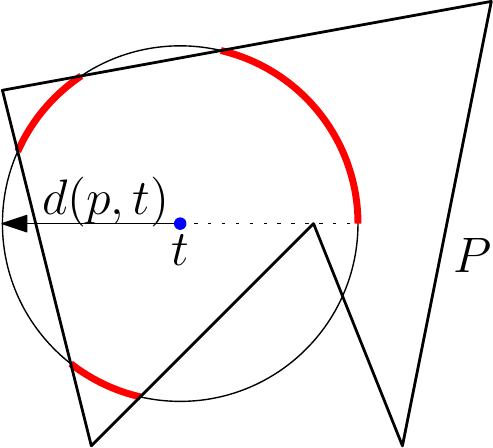}}
\hspace{0.1\textwidth}
\subfigure[]{%
		\label{fig:example2}%
		\includegraphics[width=0.22\textwidth]{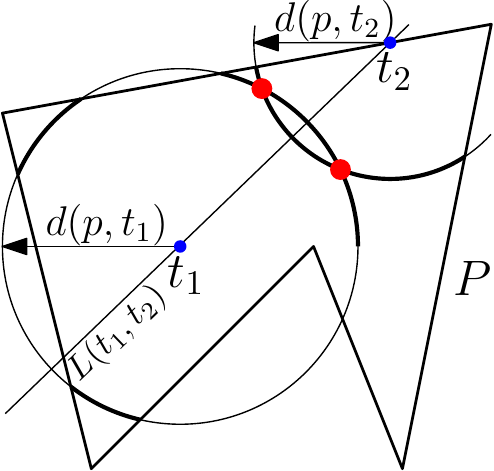}}
\hspace{0.1\textwidth}
\subfigure[]{%
		\label{fig:example3}%
		\includegraphics[width=0.22\textwidth]{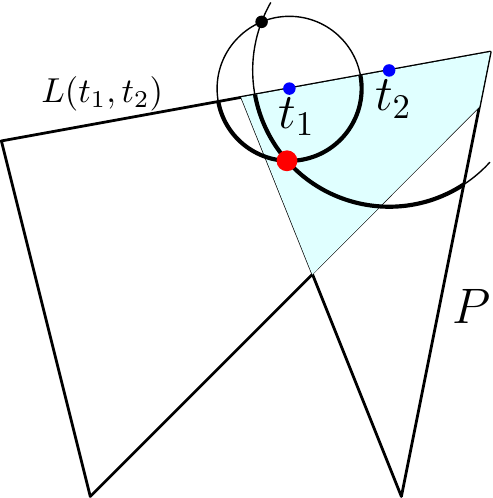}}
\caption{Trilateration example. \textbf{(a)} The point $p$ can be anywhere on $C(t,d(p,t)) \cap V(t)$, which is highlighted in red. \textbf{(b)} Ambiguity along the line $L(t_1, t_2)$. \textbf{(c)} If the map of $P$ is known then the location of $p$ can be identified precisely. The kernel of $P$ is highlighted in cyan.}
\label{fig:example}
\end{figure}

For a simple polygon $P$ in general position, we can partition it into star-shaped polygons $P_1, P_2, \ldots P_l$ such that $kernel(P_i)$, for every $1 \leq i \leq l$, does not degenerate into a single point. In every $P_i$ ($1 \leq i \leq l$) we can position a pair of towers on a line segment in $kernel(P_i) \cap \partial P_i$ (such that the towers belong to the same edge of $P_i$) or a triple of towers in $kernel(P_i)$ if $kernel(P_i) \cap \partial P_i$ is empty or contains a single point. Notice that a pair of towers positioned on the edge of $P_i$ will not necessarily be on the boundary of $P$. Thus, to localize an agent, it is not enough to know the map of $P$. We need to know more, for example, if in addition to the map of $P$ we know the partition of $P$ into star-shaped polygons and which pair of towers is responsible for which subpolygon then the agent can be localized.

In our solution we do not use this extra information or the map of $P$. Moreover, to get a tight bound of $\lfloor 2n/3\rfloor$ towers, we abstain from placing a triple of towers per subpolygon, since some polygons cannot be partitioned into less than $\lfloor n/3\rfloor$ star-shaped subpolygons. The idea is to use a \emph{parity trick}.

\vspace{10pt}
\textbf{Parity trick:} Let $L(u, v)$ be the line through points $u$ and $v$. Let $L(u, v)^+$ denote the half plane to the left of $L(u,v)$ (or above, if $L$ is horizontal). Similarly, $L(u, v)^-$ denotes the half plane to the right (or below) of $L$. We embed information about the primary orientation of the pair of towers into their coordinates. If we want a pair $t_1$, $t_2$ of towers to be responsible for $L(t_1, t_2 )^+$ (respectively $L(t_1, t_2 )^-$), then we position the towers at a distance which is a reduced rational number whose numerator is even (respectively odd). In this way, we specify on which side of $L(t_1 , t_2 )$ the primary localization region of $t_1$ and $t_2$ resides. Refer to Sect~\ref{subsec:PartitionAlg} where the parity trick is explained in greater detail.

To achieve localization with at most $\lfloor 2n/3\rfloor$ towers we should partition $P$ into at most $\lfloor n/3\rfloor$ star-shaped polygons $P_1,\ldots, P_l$ such that there exists a line segment $\overline{uv} \in kernel(P_i) \cap \partial P_i$ such that $P_i \in L(u, v )^+$ or $P_i \in L(u, v )^-$ for every $1 \leq i \leq l$ (we assume that $u$ and $v$ are distinct points).

\begin{theorem}[Chv{\'a}tal's Theorem~\cite{Chvatal197539}]
\label{thm:chvatal}
Every triangulation of a polygon with $n$ vertices can be partitioned into $m$ fans where $m \leq \lfloor n/3\rfloor$.
\end{theorem}

The statement of the following lemma may seem trivial, still we provide its proof for completeness. 

\begin{lemma}
\label{lem:pentagon}
Any simple polygon $P$ with $3$, $4$ or $5$ sides is star-shaped and its kernel contains a boundary segment that is not a single point.
\end{lemma}
\begin{proof}
Let $n$ be the number of vertices of $P = v_1 v_2 \ldots v_n$. By Theorem~\ref{thm:chvatal}, $P$ can be partitioned into $\lfloor n/3\rfloor = 1$ fans (since $n = 3$, $4$ or $5$). Notice that a fan is star-shaped by definition, from which $P$ is star-shaped. The kernel of $P$ is an intersection of $n$ half-planes defined by the edges of $P$. Let $r$ be the number of reflex angles of $P$. There are three cases to consider:
\begin{enumerate}
\item $r = 0$: $P$ is a convex polygon and thus $kernel(P) = P$, implying $\partial P \in kernel(P)$.
\item $r = 1$: Let $v_1$ be the vertex of $P$ at the reflex angle. Refer to Fig.~\ref{fig:KernelBoundary1} and~\ref{fig:KernelBoundary2} for possible polygons. The angles at $v_2$ and $v_n$ are smaller than $\pi$ by the simplicity of $P$. The claim of the lemma then is implied by inspection. Consider the edge $\overline{v_3 v_4}$. It contains a line segment that belongs to $kernel(P)$.
\begin{figure}[h]
\centering
\subfigure[]{%
		\label{fig:KernelBoundary1}%
		\includegraphics[width=0.16\textwidth]{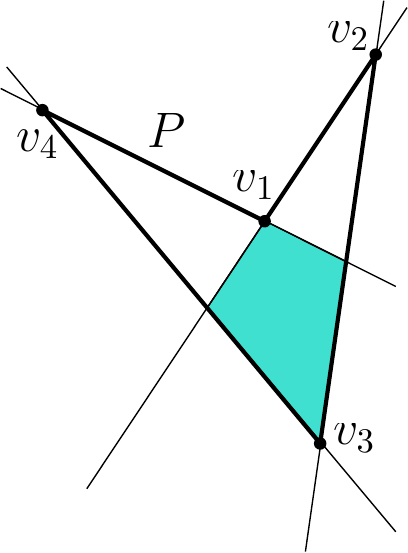}}
\hspace{0.05\textwidth}
\subfigure[]{%
		\label{fig:KernelBoundary2}%
		\includegraphics[width=0.16\textwidth]{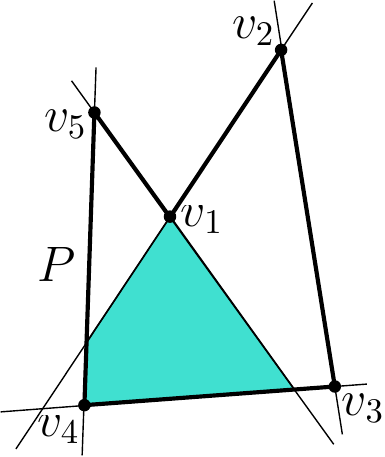}}
\hspace{0.05\textwidth}
\subfigure[]{%
		\label{fig:KernelBoundary3}%
		\includegraphics[width=0.16\textwidth]{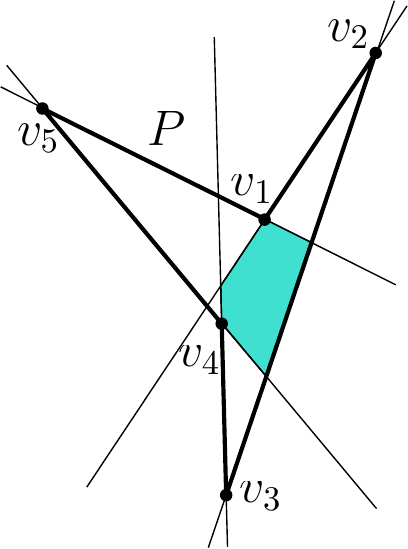}}
\hspace{0.05\textwidth}
\subfigure[]{%
		\label{fig:KernelBoundary4}%
		\includegraphics[width=0.19\textwidth]{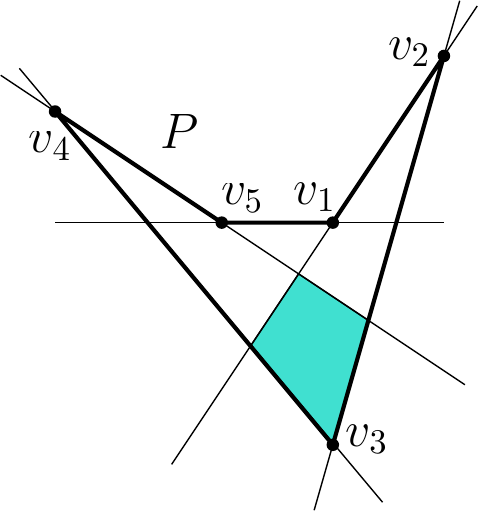}}
\caption{The kernel of $P$ is highlighted in cyan.}
\label{fig:KernelBoundary}
\end{figure}
\item $r = 2$: In this case $n = 5$. Let $v_1$ be the vertex of $P$ at one of the reflex angles. Refer to Fig.~\ref{fig:KernelBoundary3} and~\ref{fig:KernelBoundary4} for possible polygons. The vertex at the other reflex angle is either adjacent to $v_1$ or not. Consider first the case where the vertices at the reflex angles of $P$ are not adjacent. Without loss of generality let $v_4$ be a vertex of $P$ at reflex angle (refer to Fig.~\ref{fig:KernelBoundary3}). The angles at $v_2$, $v_3$ and $v_5$ are smaller than $\pi$ by the simplicity of $P$. It follows that the edge $\overline{v_2 v_3}$ must contain a line segment that belongs to $kernel(P)$.
Consider now the case where the vertices at the reflex angles of $P$ are adjacent. Without loss of generality let $v_5$ be a vertex at reflex angle of $P$ (refer to Fig.~\ref{fig:KernelBoundary4}). The claim of the lemma holds similarly to the discussion presented in the case for $r = 1$, Fig.~\ref{fig:KernelBoundary1}.
\end{enumerate}
\qed
\end{proof}

The problem we study is twofold:
\begin{enumerate}
\item We are given a simple polygon $P$ of size $n$. Our goal is to position at most $\lfloor 2n/3\rfloor$ towers inside $P$ such that every point $p \in P$ can be localized.
\item  We want to design a localization algorithm which does not know $P$, but knows that the locations of the towers were computed using the parity trick. For any point $p \in P$, its location can be found by using the coordinates of the towers that \textbf{see} $p$ and the distances from those towers to $p$.
\end{enumerate}

It may seem counter-intuitive, but the knowledge of the parity trick is \textbf{stronger} than the knowledge of the map of $P$. Some towers (while still on the boundary of some subpolygon of the partition) may end up in the interior of $P$. This is not a problem when the parity trick is used but may lead to ambiguities when only the map of $P$ is known (refer for example to Fig.~\ref{fig:example2}).

The following theorem shows that additional information like the parity trick or the map of $P$ (including its partition) is necessary to achieve localization with the use of less than $n$ towers.

\begin{theorem}
\label{thm:no_map}


Let $P$ be a simple polygon with $n$ vertices. An agent cannot localize itself inside $P$ when less than $n-(n\bmod 3)$ towers are used if the {\em only} information available to the localization algorithm are the coordinates of the towers and the distances of the agent to the towers visible to it. 
\end{theorem}
\begin{proof}
Let $P$ be an arbitrary simple polygon with $n=3$ vertices. 
Assume to the contrary, that $P$ can be trilaterated with $2$ towers. Given the coordinates of the two towers $t_1$ and $t_2$ together with the distances to a query point $p$ one can deduce that $p$ is in one of the two possible locations $C(t_1,d(p,t_1)) \cap C(t_2,d(p,t_2)) = \{p_1, p_2\}$. But without additional information it is impossible to choose one location over another. Refer to Fig.~\ref{fig:no_map}. 
Similarly, any quadrilateral or pentagon requires at least $3$ towers to trilaterate it. 

\begin{figure}[h]
\centering
\subfigure[]{%
		\label{fig:no_map1}%
		\includegraphics[width=0.24\textwidth]{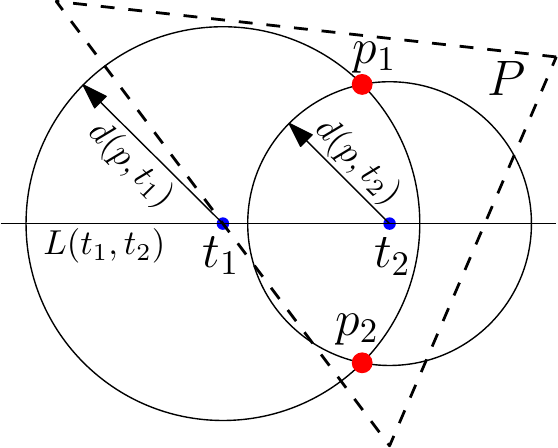}}
\hspace{0.1\textwidth}
\subfigure[]{%
		\label{fig:no_map2}%
		\includegraphics[width=0.27\textwidth]{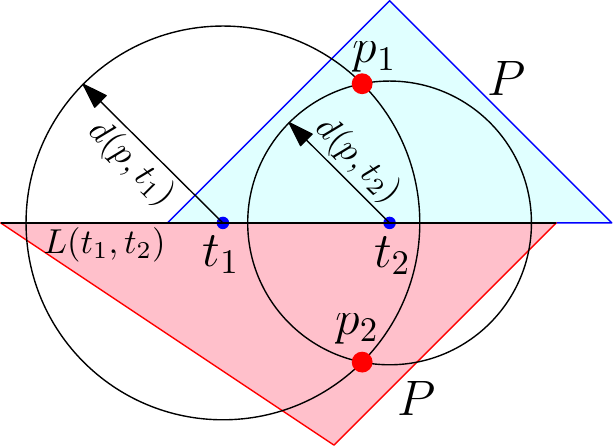}}
\caption{How to locate $p$ if the map of the polygon is unavailable? \textbf{(a)} Ambiguity along $L(t_1, t_2)$. We know that $p_1$ and/or $p_2$ belongs to $P$. \textbf{(b)} We need additional information to tell if $P$ (and $p$) is above or below $L(t_1, t_2)$.}
\label{fig:no_map}
\end{figure}

Let $P$ be a comb polygon with $n = 3k + q$ vertices (for integer $k \geq 1$ and $q = 0$, $1$ or $2$) such that one of its $k$ spikes contains $q$ extra vertices. Refer to Fig.~\ref{fig:comb_full}. No point of $P$ can see the complete interior of two different spikes.  
Assume to the contrary that $P$ can be trilaterated with less than $n - q$ towers. In other words, assume that $P$ can be trilaterated with $3k - 1$ towers. It follows that one of the spikes contains less than $3$ towers. We showed for smaller polygons (with $n = 3$, $4$ or $5$) that two towers cannot trilaterate it. Even if we know that the two towers are positioned on the same edge of $P$, the polygon can be mirrored along this edge to avoid unique trilateration. Observe that no polygon or spike can be trilaterated with one or no towers. 

\begin{figure}[h]
\centering
\subfigure[]{%
		\label{fig:comb}%
		\includegraphics[width=0.42\textwidth]{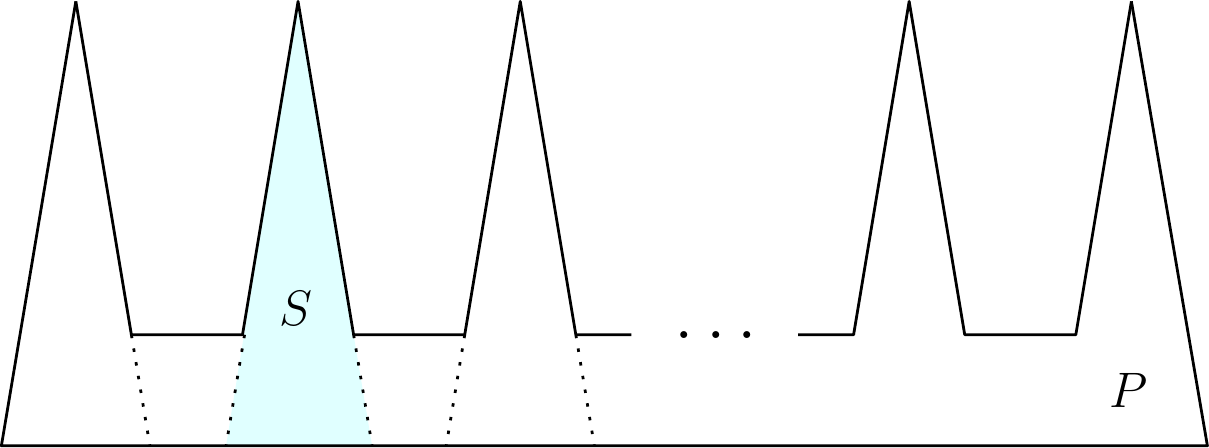}}
\hspace{0.1\textwidth}
\subfigure[]{%
		\label{fig:comb1}%
		\includegraphics[width=0.077\textwidth]{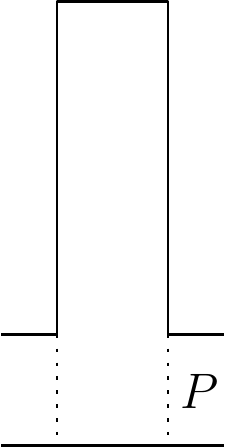}}
\hspace{0.1\textwidth}
\subfigure[]{%
		\label{fig:comb2}%
		\includegraphics[width=0.077\textwidth]{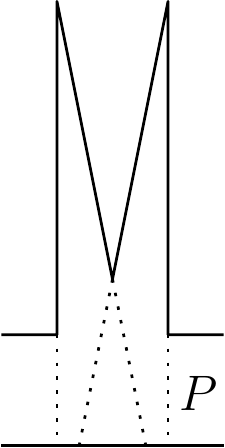}}
\caption{Comb polygon with $n = 3k + q$ vertices. \textbf{(a)} One of the $k$ spikes $S$ is shown in cyan. \textbf{(b)} One of the spikes when $q = 1$. \textbf{(c)} One of the spikes when $q = 2$.}
\label{fig:comb_full}
\end{figure}
\qed
\end{proof}

\section{T{\'o}th's Partition}
\label{sec:partition}
\pdfbookmark[1]{T{\'o}th's Partition}{sec:partition}

A simple polygonal chain $c$ is a \emph{cut} of a polygon $P$ if the two endpoints of $c$ belong to $\partial P$ and all interior points of $c$ are interior points of $P$. A cut $c$ decomposes $P$ into two polygons $P_1$ and $P_2$ such that $P = P_1 \cup P_2$ and $c = P_1 \cap P_2$. 

\begin{definition}[Diagonal] 
A \emph{diagonal} of $P$ is a line segment that connects non-adjacent vertices of $P$ and is contained in or on the boundary of $P$. If $c$ is a cut and a diagonal of $P$ then it is called a \emph{diagonal cut}.
\end{definition}

Notice that the above definition of the diagonal is non-standard. We define a diagonal as a line segment that can contain boundary points in its interior, while a diagonal in a standard definition is a line segment whose intersection with the boundary is precisely its two endpoints.

T{\'o}th showed in~\cite{Toth2000121} that any simple polygon in general position of size $n$ can be guarded by $\lfloor n/3\rfloor$ point guards whose range of vision is $180^\circ$ (let us call this result \textbf{T{\'o}th's Theorem}). His approach is to decompose $P$ into subpolygons via \emph{cuts} and to specify the locations of the guards. The cuts are composed of one or two line segments and are not restricted to be diagonal cuts. He uses a constructive induction to show his main result.
Let $n_1$ (respectively $n_2$) denote the size of $P_1$ (respectively $P_2$). When $P$ contains a cut that satisfies $\left \lfloor \frac{n_1}{3} \right \rfloor + \left \lfloor \frac{n_2}{3} \right \rfloor \leq \left \lfloor \frac{n}{3} \right \rfloor$, the proof of T{\'o}th's Theorem can be completed by applying induction to both $P_1$ and $P_2$.

T{\'o}th's method heavily relies upon the partitioning of a polygon into subpolygons (on which he can apply induction). He performs diagonal cuts whenever it is possible, otherwise he cuts along a continuation of some edge of $P$; along a two-line segment made of an extension of two edges of $P$ that intersect inside $P$; or along the bisector of a reflex vertex of $P$. Notice that the three latter types of cuts may introduce new vertices that are not necessarily in general position with the given set of vertices.

However, T{\'o}th assumes that every cut produces polygons in general position, which is a very strong assumption. We strengthen the work~\cite{Toth2000121} by lifting this assumption and reproving T{\'o}th's result. We assume that the input polygon is in general position, while non-general position may occur for
subpolygons of the partition. Moreover, we found and fixed several mistakes in~\cite{Toth2000121}.

\begin{definition}[Good Cut] 
Let $n > 5$. A cut is called a \emph{good cut} if it satisfies the following: $\left \lfloor \frac{n_1}{3} \right \rfloor + \left \lfloor \frac{n_2}{3} \right \rfloor \leq \left \lfloor \frac{n}{3} \right \rfloor$. If a good cut is a diagonal then it is called a \emph{good diagonal cut}.
\end{definition}

\begin{definition}[Polygon*] 
\label{def:polygon}
\emph{Polygon*} is a weakly simple polygon (as defined in~\cite{DBLP:conf/soda/ChangEX15}) whose vertices are distinct and whose interior angles are strictly bigger than $0$ and strictly smaller than $2\pi$. Notice that polygon* includes simple polygons.
\end{definition}

We assume that every polygon of the partition (i.e. subpolygon of $P$) is polygon*. To avoid confusion, let $P'$ refer to a polygon* to which we apply the inductive partitioning. In other words, $P'$ can refer to the input polygon $P$ before any partition was applied as well as to any subpolygon of $P$ obtained during the partitioning of $P$. Recall that $P$ is in general position, while $P'$ may not be in general position.  

\begin{definition}[Triangulation] 
A decomposition of a polygon into triangles by a maximal
set of non-intersecting (but possibly overlapping) diagonals is called a \emph{triangulation} of the polygon. Notice that a triangulation of polygon* may contain triangles whose three vertices are collinear.
\end{definition}

Notice that the above definition of triangulation is different from the classical one by that that it allows overlapping diagonals and thus permits triangles with three collinear vertices.

If the polygon is not in general position then its triangulation may include triangles whose three vertices are collinear. We call such triangles \emph{degenerate triangles}. Refer to Fig.~\ref{fig:case01}, showing an example of a degenerate triangle $\triangle v_1 v_2 v_3$. The diagonal $\overline{v_1 v_3}$ cannot be a good diagonal cut even if it partitions $P'$ into $P_1$ and $P_2$ such that $\left \lfloor \frac{n_1}{3} \right \rfloor + \left \lfloor \frac{n_2}{3} \right \rfloor \leq \left \lfloor \frac{n}{3} \right \rfloor$, because interior points of a cut cannot contain points of $\partial P'$. 

To extend T{\'o}th's partition to polygons* 
we need to extend the definition of a cut. A simple polygonal chain $c'$ is a \emph{dissection} of $P'$ if the two endpoints of $c'$ belong to $\partial P'$ and all interior points of $c'$ are either interior points of $P'$ or belong to $\partial P'$. A dissection $c'$ decomposes $P'$ into two polygons* $P_1$ and $P_2$ (that are not necessarily in general position) such that $P' = P_1 \cup P_2$ and $c' = P_1 \cap P_2$. If $c'$ is a dissection and a diagonal of $P'$ then it is called a \emph{diagonal dissection}.

\begin{definition}[Good Dissection] 
Let $n > 5$. A dissection is called a \emph{good dissection} if $\left \lfloor \frac{n_1}{3} \right \rfloor + \left \lfloor \frac{n_2}{3} \right \rfloor \leq \left \lfloor \frac{n}{3} \right \rfloor$.
\end{definition}

We extend the results in~\cite{Toth2000121} by removing the assumption that the partitioning produces subpolygons in general position and by thus strengthening the result. In this paper, we need to refer to many Lemmas, Propositions and Claims from~\cite{Toth2000121}. Indeed, we apply some of these results, we fill the gaps in some proofs of these results and we generalize some others. In order for this paper to be self-contained, we write the statements of all these results with their original numbering~\cite{Toth2000121}. 

\medskip
\noindent
\textbf{Simplification Step:} If $P'$ has consecutive vertices $v_1$, $v_2$ and $v_3$ along $\partial P'$ that are collinear then: 
\begin{enumerate}
\item If the angle of $P'$ at $v_2$ is $\pi$ then replace $v_2$ together with its adjacent edges by the edge $\overline{v_1 v_3}$.
\item If the angle of $P'$ at $v_2$ is $0$ or $2\pi$ then delete $v_2$ together with its adjacent edges and add the edge $\overline{v_1 v_3}$. Assume w.l.o.g. that the distance between $v_3$ and $v_2$ is smaller than the distance between $v_1$ and $v_2$. The line segment $\overline{v_2 v_3}$ will be guarded despite that it was removed from $P'$. Refer to the following subsections for examples.
\end{enumerate} 
In both cases we denote the updated polygon by $P'$; its number of vertices decreased by $1$. 

Assume for simplicity that $P'$ does not contain consecutive collinear vertices along $\partial P'$. Let $n > 5$ be the number of vertices of $P'$. Let $k$ be a positive integer and $q \in \{0, 1, 2\}$ such that $n = 3k +q$. Notice that a diagonal dissects $P'$ into $P_1$ (of size $n_1$) and $P_2$ (of size $n_2$) such that $n_1 + n_2 = n+2$.

The proofs of the following three propositions are identical to those in T{\'o}th's paper.

\begin{mydefP}
\label{TothProp1}
For $P'$ with $n = 3k$ any diagonal is a good 
dissection. 
\end{mydefP}

\begin{mydefP}
\label{TothProp2}
For $P'$ with $n = 3k+1$ there exists a diagonal that represents a good dissection.
\end{mydefP}

\begin{mydefP}
\label{TothProp3}
For $P'$ with $n = 3k+2$ a dissection is a good dissection if it decomposes $P'$ into polygons* $P_1$ and $P_2$ (not necessarily simple polygons) such that $n_1 = 3k_1 + 2$ and $n_2 = 3k_2 + 2$ (for $k_1 +k_2 = k$).
\end{mydefP}

\subsection{Case Study}
\label{subsec:CaseStudy}
\pdfbookmark[2]{Case Study}{subsec:CaseStudy}

\begin{wrapfigure}{r}{0.35\textwidth}
\vspace{-20pt}
\centering
\includegraphics[width=0.3\textwidth]{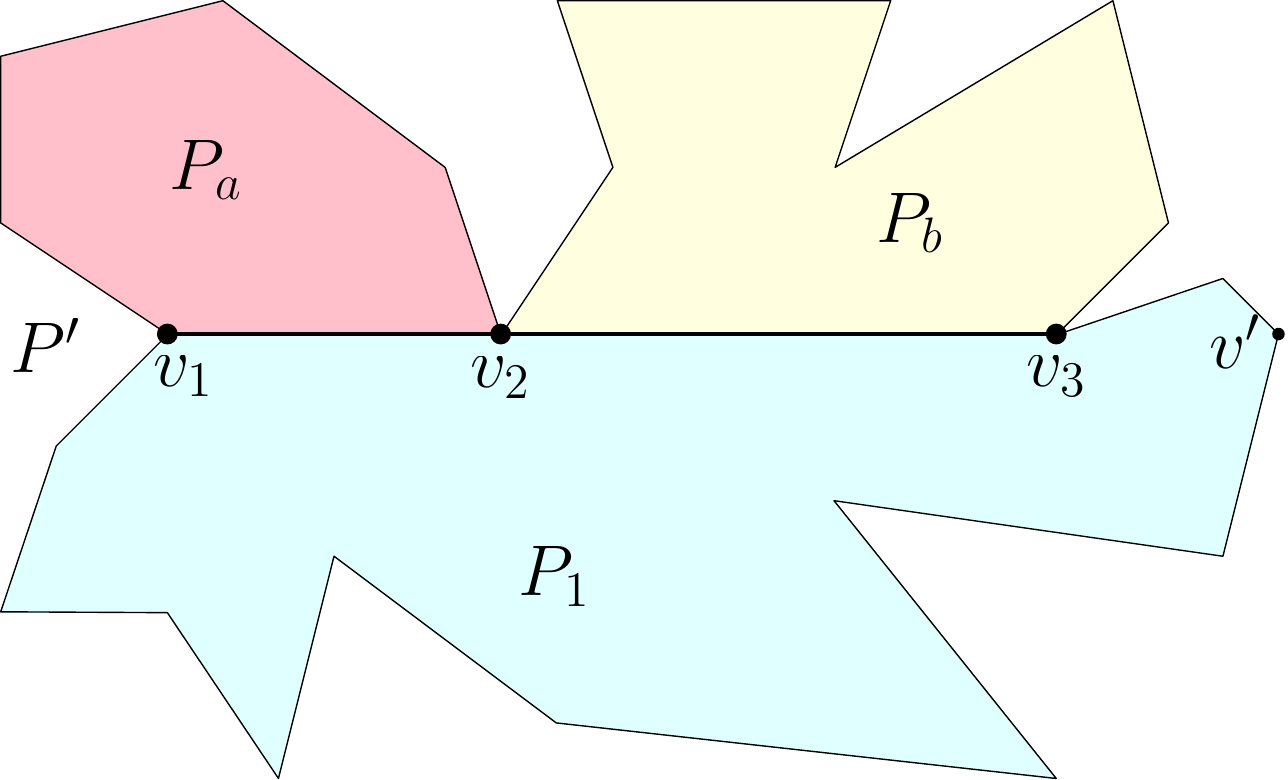}
\caption{The diagonal $\overline{v_1 v_3}$ contains a vertex $v_2$ of $P'$. $P_2 = P_a \cup P_b$. Notice that $v_2 \notin P_1$.}
\label{fig:case01}
\vspace{-20pt}
\end{wrapfigure}
In this subsection we study how Propositions~\ref{TothProp1},~\ref{TothProp2} and~\ref{TothProp3} can be applied to polygons*. 
Let $\triangle v_1 v_2 v_3$ be a minimum degenerate triangle in $P'$, i.e. it does not contain any vertex of $P'$ other than $v_1$, $v_2$ or $v_3$. Consider the example depicted in Fig.~\ref{fig:case01}. The diagonal $\overline{v_1 v_3}$ partitions $P'$ into two polygons: $P_1$ and $P_2$ ($P_2$ can be further viewed as a union of two subpolygons $P_a$ and $P_b$). Notice that $v_2 \notin P_1$. There is a possibility for $P'$ to have a vertex $v' \notin \{v_1, v_2, v_3\}$ such that $v' \in L_{v_1 v_3}$. However $v'$ cannot belong to the line segment $\overline{v_1 v_3}$, otherwise, the triangle $\triangle v_1 v_2 v_3$ is not minimum, and we can choose $\triangle v_1 v_2 v'$ or $\triangle v' v_2 v_3$ instead. It follows that no vertex of $P_1$ (respectively $P_a$, $P_b$) is located on the edge $\overline{v_1 v_3}$ (respectively $\overline{v_1 v_2}$, $\overline{v_2 v_3}$). Notice that $P'$ may contain an edge $e'$ that contains $\overline{v_1 v_3}$; $P_1$ inherits it (since $v_2 \notin P_1$), but it won't cause any trouble because of the ``simplification'' step described in the beginning of this section. Notice also, that two vertices of $P'$ cannot be at the same location. 


Assume that $\overline{v_1 v_3}$ is a good dissection, i.e. it decomposes $P'$ into $P_1$ and $P_2$ such that $\left \lfloor \frac{n_1}{3} \right \rfloor + \left \lfloor \frac{n_2}{3} \right \rfloor \leq \left \lfloor \frac{n}{3} \right \rfloor$. Assume also that $n>5$. Let $n_a$ (respectively $n_b$) be the number of vertices of $P_a$ (respectively $P_b$). Notice that $n_a + n_b = n_2 + 1$ because $v_2$ was counted twice. When it is possible, we prefer to avoid cutting along the diagonal $\overline{v_1 v_3}$. However, when necessary, it can be done in the following way. We consider three cases (refer to Fig.~\ref{fig:case01}):

{\color{red} \textbf{Case 1:}} $n = 3k$. By T{\'o}th's Proposition~\ref{TothProp1} every diagonal is a good dissection. If $P_a$ is \textbf{not} composed of the line segment $\overline{v_1 v_2}$ only, then $\overline{v_1 v_2}$ is a good diagonal dissection, that partitions $P'$ into two polygons $P_a$ and $P_1 \cup P_b$. Otherwise, $\overline{v_2 v_3}$ is a good diagonal dissection, that partitions $P'$ into two polygons $P_b$ and $P_1 \cup P_a$. Notice that $\overline{v_1 v_2}$ and $\overline{v_2 v_3}$ cannot be both edges of $P'$ because of the simplification step we applied to $P'$.

Alternatively, we can dissect $P'$ along $\overline{v_1 v_3}$. The polygon $P_2$ has an edge $\overline{v_1 v_3}$ that contains vertex $v_2$. If a $180^\circ$-guard is located at $v_2$ for any subpolygon $P_2'$ of $P_2$, then we position the towers close to $v_2$ on a line segment that contains $v_2$ (but not necessarily in $kernel(P_2') \cap \partial P_2'$) only if $v_2$ is not a vertex of the original polygon $P$. However, if $v_2 \in P$ then we position our towers in the close proximity to $v_2$ but in the interior of $P_1$, oriented towards $P_2$.

{\color{red} \textbf{Case 2:}} $n = 3k+1$. Since $\overline{v_1 v_3}$ decomposes $P'$ into $P_1$ and $P_2$ such that $\left \lfloor \frac{n_1}{3} \right \rfloor + \left \lfloor \frac{n_2}{3} \right \rfloor \leq \left \lfloor \frac{n}{3} \right \rfloor$ then either $n_1 = 3k_1 + 1$ and $n_2 = 3k_2 + 2$, or $n_1 = 3k_1 + 2$ and $n_2 = 3k_2 + 1$ for $k_1 +k_2 = k$. If $\overline{v_1 v_2}$ is an edge of $P'$ (meaning that $P_a$ consists of a line segment $\overline{v_1 v_2}$ only) then we can dissect $P'$ either along $\overline{v_2 v_3}$ or along $\overline{v_1 v_3}$. In the former case $P'$ will be decomposed into $P_1 \cup \{v_2\}$ and $P_b$ ($v_2$ will be deleted from $P_1 \cup \{v_2\}$ during the simplification step to obtain $P_1$). In the latter case $P'$ will be decomposed into $P_1$ and $P_b \cup \{v_1\}$ ($v_1$ will be removed from $P_b \cup \{v_1\}$ during the simplification step to obtain $P_b$). In either case we do not have to specifically guard the line segment $\overline{v_1 v_2}$. If $P_1$ is guarded then so is $\overline{v_1 v_2}$. Similarly treat the case where $\overline{v_2 v_3}$ is an edge of $P'$.

Assume that $n_a, n_b >2$. Recall, that $n_b = n_2 - n_a + 1$. Several cases are possible:
\begin{itemize}
\item[$\blacktriangleright$] $n_a = 3k_a + 1$. Then we dissect $P'$ along $\overline{v_1 v_2}$. This partition creates the polygons $P_a$ of size $n_a = 3k_a + 1$ and $P_1 \cup P_b$ of size $3(k_1 + k_2 - k_a) +2$. 
\item[$\blacktriangleright$] $n_a = 3k_a + 2$. Then we dissect $P'$ along $\overline{v_1 v_2}$. This partition creates the polygons $P_a$ of size $n_a = 3k_a + 2$ and $P_1 \cup P_b$ of size $3(k_1 + k_2 - k_a) +1$. 
\item[$\blacktriangleright$] $n_a = 3k_a$. Then:
\begin{itemize}
\item $n_1 = 3k_1 + 2$ and $n_2 = 3k_2 + 1$. Then we dissect $P'$ along $\overline{v_2 v_3}$. This partition creates the polygons $P_b$ of size $n_b = 3(k_2 - k_a)+2$ and $P_1 \cup P_a$ of size $3(k_1 +k_a) +1$. 
\item $n_1 = 3k_1 + 1$ and $n_2 = 3k_2 + 2$. Then both $\overline{v_1 v_2}$ and $\overline{v_2 v_3}$ are \textbf{not} good diagonal cuts. We cannot avoid dissecting along $\overline{v_1 v_3}$. If a $180^\circ$-guard eventually needs to be positioned at $v_2$ (for any subpolygon of $P_2$) and $v_2$ is a vertex of $P$, then we position our towers close to $v_2$ but in the interior of $P_1$, oriented towards~$P_2$.
\end{itemize}
\end{itemize}

{\color{red} \textbf{Case 3:}} $n = 3k+2$. In this case $n_1 = 3k_1+2$ and $n_2 = 3k_2+2$. If $n_a = 2$ then either dissect $P'$ along $\overline{v_2 v_3}$ (and later delete $v_2$ from $P_1 \cup v_2$ during the simplification step), or dissect $P'$ along $\overline{v_1 v_3}$ (delete $v_1$ from $P_b \cup v_1$ during the simplification step). Notice that in both cases $\overline{v_1 v_2}$ will be guarded. The case where $n_b = 2$ is similar.

If $n_a, n_b >2$ then we have no choice but to dissect along $\overline{v_1 v_3}$. This creates a polygon $P_2$ whose kernel degenerates into a single point $v_2$. It is not a problem for $180^\circ$-guards but it is a serious obstacle for our problem. If $v_2$ is a vertex of $P$ and a $180^\circ$-guard is located at $v_2$ for any subpolygon of $P_2$ then we position our towers close to $v_2$ but in the interior of $P_1$, and orient those towers towards $P_2$. We consider the general approach to this problem in the following subsection.

\subsection{Point-Kernel Problem}
\label{subsec:PointKernel}
\pdfbookmark[2]{Point-Kernel Problem.}{subsec:PointKernel}

In Section~\ref{subsec:CaseStudy} we studied simple cases where a diagonal dissection is applied to polygons* (or polygons in non-general position). In this subsection, we show how to circumvent some difficulties that arise when adapting T{\'o}th's partitioning to our problem.

The dissection of $P$ may create subpolygons that are polygons*. 
This means that the partition of $P$ may contain star-shaped polygons whose kernels degenerate into a single point. While this is not a problem for $180^\circ$-guards, it is a serious obstacle for tower positioning. Indeed, we need at least two distinct points in the kernel of each part of the partition to trilaterate~$P$. 

\begin{wrapfigure}{r}{0.42\textwidth}
\centering
\includegraphics[width=0.38\textwidth]{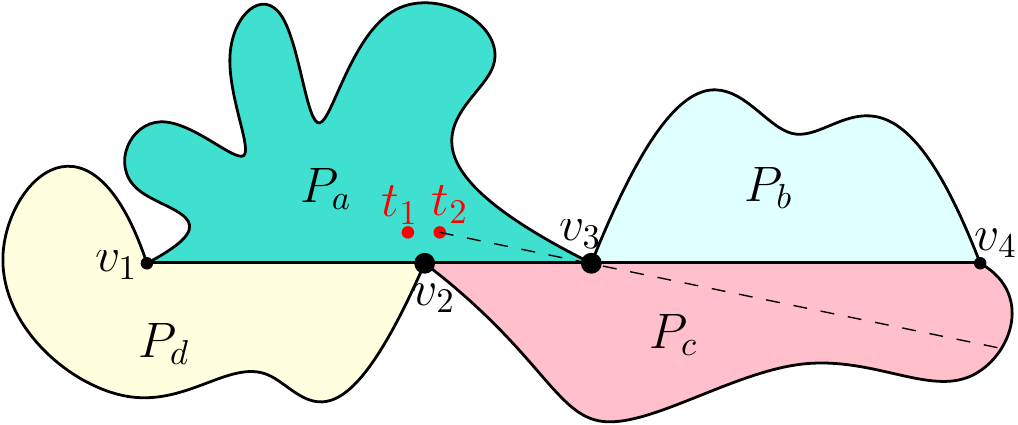}
\caption{$\overline{v_1 v_4}$ is a good diagonal dissection that contains two vertices of $P$: $v_2$ and $v_3$. $P_1 = P_a\cup P_b$ and $P_2 = P_c\cup P_d$. Notice that $v_2 \notin P_a$ and $v_3 \notin P_c$.}
\label{fig:PointKernel1}
\vspace{-10pt}
\end{wrapfigure}

Let $\overline{v_1 v_4}$ be a good dissection, i.e. it decomposes $P'$ into $P_1$ and $P_2$ such that $\left \lfloor \frac{n_1}{3} \right \rfloor + \left \lfloor \frac{n_2}{3} \right \rfloor \leq \left \lfloor \frac{n}{3} \right \rfloor$. Assume also that $n>5$. Assume that $\overline{v_1 v_4}$ contains at least $2$ more vertices of $P'$ in its interior. Let $v_2, v_3 \in \overline{v_1 v_4}$. Refer to Fig.~\ref{fig:PointKernel1}. Notice that $v_2$ and $v_3$ belong to subpolygons of $P'$ on the opposite sides of $L(v_1, v_4)$. Recall that the original polygon $P$ is given in general position. Thus at most two vertices of $P$ belong to the same line. If $v_2$ and $v_3$ are not both vertices of $P$ then there is no problem for tower positioning, and the dissection via $\overline{v_1 v_4}$ can be applied directly. We then follow T{\'o}th's partition and replace every $180^\circ$-guard with a pair of towers in the guard's vicinity. However, if $v_2$ and $v_3$ are vertices of $P$ then the visibility of towers can be obstructed. Consider the example of Fig.~\ref{fig:PointKernel1}. Assume that a $180^\circ$-guard is positioned at $v_2$ to observe a subpolygon of $P_c \cup P_d$ (say, pentagon $v_1 v_4 v_c v_2 v_d$ for $v_c \in P_c$ and $v_d \in P_d$). We cannot replace the $180^\circ$-guard with a pair of towers because the kernel of the pentagon degenerates into a single point $v_2$. Our attempt to position towers in the vicinity of $v_2$ but interior to $P_a$ is not successful either, because $v_3$, as a vertex of $P$, blocks visibility with respect to $P_c$. We have to find another dissection. In general, we are looking for a dissection that destroys the diagonal $\overline{v_1 v_4}$.

Recall that in this subsection $P'$ is a subpolygon of $P$ with collinear vertices $v_1$, $v_2$, $v_3$ and $v_4$. Assume that $v_2$ and $v_3$ are vertices of $P$. It follows that there are no vertices of $P$ on $L(v_1, v_4)$ other than $v_2$ and $v_3$. Assume that $\overline{v_1 v_4}$ is a good dissection of $P'$. If one of $\overline{v_1v_3}$, $\overline{v_1v_2}$, $\overline{v_3v_4}$, $\overline{v_2v_4}$ or $\overline{v_2v_3}$ is a good dissection then we cut along it. Only the cut along $\overline{v_2v_3}$ destroys the diagonal $\overline{v_1v_4}$ and eliminates the problem of having towers in the vicinity of $v_2$ (respectively $v_3$) with visibility obstructed by $v_3$ (respectively $v_2$). The cut along $\overline{v_1v_3}$, $\overline{v_1v_2}$, $\overline{v_3v_4}$ or $\overline{v_2v_4}$ only reduces $P'$. So assume that none of the diagonals $\overline{v_1v_3}$, $\overline{v_1v_2}$, $\overline{v_3v_4}$, $\overline{v_2v_4}$ or $\overline{v_2v_3}$ represent a good dissection. It follows from T{\'o}th's Proposition~\ref{TothProp1} that $n \neq 3k$. Consider the following cases:

{\color{red} \textbf{Case 1:}} $n = 3k+1$. By T{\'o}th's Proposition~\ref{TothProp2} either $n_1 = 3k_1 + 2$ and $n_2 = 3k_2 + 1$, or $n_1 = 3k_1 + 1$ and $n_2 = 3k_2 + 2$ for $k_1 +k_2 = k$, where $P_1 = P_a\cup P_b$ and $P_2 = P_c\cup P_d$. Assume w.l.o.g. that $n_1 = 3k_1 + 2$ and $n_2 = 3k_2 + 1$. We assumed that $\overline{v_2v_3}$ is not a good dissection and thus the size of $P_c \cup P_b$ is a multiple of $3$ and the size of $P_d \cup P_a$ is a multiple of $3$. Consider the following three subcases:

{\color{blue} \textbf{Case 1.1:}} $|P_d| = 3k_d$; thus $|P_c| = 3k_c+2$ or $|P_c| = 2$. Since the size of $P_c \cup P_b$ is a multiple of $3$, we have $|P_b| = 3k_b+2$ or $|P_b| = 2$. Notice that $P_b$ and $P_c$ cannot be both of size $2$, otherwise  $v_4$ would be deleted during the simplification step. Both diagonals $\overline{v_3v_4}$ and $\overline{v_2v_4}$ represent a good dissection. Thus, if $|P_c| \neq 2$ then dissect along $\overline{v_2v_4}$; if $|P_b| \neq 2$ then dissect along $\overline{v_3v_4}$. The polygon $P_a \cup P_d \cup \triangle v_2 v_3 v_4$ will be simplified and it will loose $\overline{v_3 v_4}$ (which is guarded by $P_b$ or/and $P_c$).

{\color{blue} \textbf{Case 1.2:}} $|P_d| = 3k_d+1$ or $|P_d| = 3k_d+2$. Then $\overline{v_1v_2}$ is a good diagonal dissection.

{\color{blue} \textbf{Case 1.3:}} $|P_d| = 2$. Then $|P_a| = 3k_a +2$ and thus $\overline{v_1v_3}$ is a good diagonal dissection, which creates the polygons $P_a$ and $P_b \cup P_c \cup \{v_1\}$. The vertex $v_1$ will be deleted from $P_b \cup P_c \cup \{v_1\}$ during the simplification step. Notice that the line segment $\overline{v_1v_2}$ will be guarded by $P_a$.

{\color{red} \textbf{Case 2:}} $n = 3k+2$. Since $\overline{v_1 v_4}$ is a good dissection, it partitions $P'$ into $P_1 = P_a\cup P_b$ of size $n_1 = 3k_1 +2$ and $P_2 = P_c\cup P_d$ of size $n_2 = 3k_2 +2$.

{\color{blue} \textbf{Case 2.1:}} $|P_b| = 3k_b+1$; thus $|P_a|=3k_a +2$ or $|P_a| = 2$. We assumed that $\overline{v_1v_3}$ is not a good dissection and thus 
\begin{wrapfigure}{r}{0.42\textwidth}
\vspace{-15pt}
\centering
\includegraphics[width=0.4\textwidth]{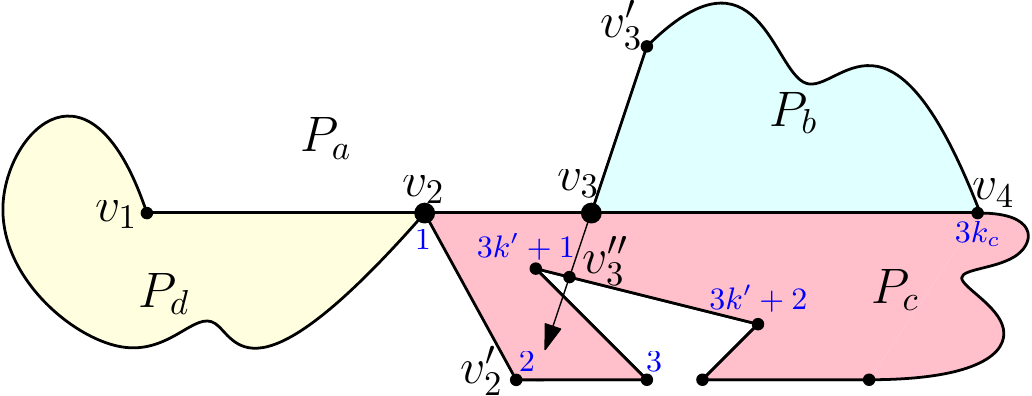}
\caption{$\overline{v_1 v_4}$ is a good diagonal dissection that contains two vertices of $P$: $v_2$ and $v_3$. $P_a = \{v_1, v_3\}$, $|P_b| = 3k_b+1$, $|P_c| = 3k_c$ and $|P_d| = 3k_d$.}
\label{fig:PointKernel2}
\vspace{-25pt}
\end{wrapfigure}
$|P_a| \neq3 k_a +2$. It follows that $P_a$ is a line segment $\overline{v_1v_3}$. Symmetrically, if $|P_d| = 3k_d+1$ then $|P_c| = 2$. If $|P_b| = 3k_b+1$ and $|P_d| = 3k_d+1$ then we have a good dissection $\overline{v_2 v_3}$, which is a contradiction.

The case where $|P_b| = 3k_b+1$ and $|P_c| = 3k_c+1$ is not possible because in this case $P_a = \overline{v_1v_3}$ and $P_d = \overline{v_1v_2}$ and thus $v_1$ would not survive the simplification step.

Thus, if $|P_b| = 3k_b+1$ then $|P_c| = 3k_c$ and $|P_d| = 3k_d$. Refer to Fig.~\ref{fig:PointKernel2}.
Let $v_2'$ be an immediate neighbour of $v_2$ such that $v_2'$ is a vertex of $P_c$ and $v_2' \neq v_4$. If $v_2'$ and $v_3$ see each other, then $\overline{v_2' v_3}$ is a diagonal of $P'$. It dissects $P'$ into two subpolygons $P_d \cup \triangle v_2v_3v_2'$ of size $3k_d+2$ and $(P_b \cup P_c) \setminus \triangle v_2v_3v_2'$ of size $3k_b+3k_c-1 = 3(k_b+k_c-1)+2$. Assume now that $v_2'$ and $v_3$ do not see each other. Let us give numerical labels to the vertices of $P_c$ as follows: $v_2$ gets label $1$, $v_2'$ gets $2$, and so on, $v_4$ will be labelled $3k_c$. If $v_3$ can see a vertex of $P_c$ with label whose value modulo $3$ equals $2$ then dissect $P'$ along the diagonal that connects $v_3$ and that vertex. If no such vertex can be seen from $v_3$, we do the following. Let $v_3'$ be an immediate neighbour of $v_3$ such that $v_3'$ is a vertex of $P_b$ and $v_3' \neq v_4$. Let $v_3''$ be the point on $\partial P_c$ where the line supporting $\overline{v_3' v_3}$ first hits $\partial P_c$ (refer to Fig.~\ref{fig:PointKernel2}). If $v_3''$ belongs to an edge with vertices whose labels modulo $3$ equal $1$ and $2$ then dissect $P'$ along $\overline{v_3 v_3''}$. This dissection creates polygons $P_d \cup \{1,2, \ldots, 3k'+1, v_3'', v_3\}$ of size $3(k_d+k')+2$
and $(\{v_3'', 3k'+2, \ldots 3k_c\} \cup P_b) \setminus \{v_3\}$ of size $3(k_c-k'+k_b-1)+2$. 

\begin{wrapfigure}{r}{0.42\textwidth}
\vspace{-10pt}
\centering
\includegraphics[width=0.36\textwidth]{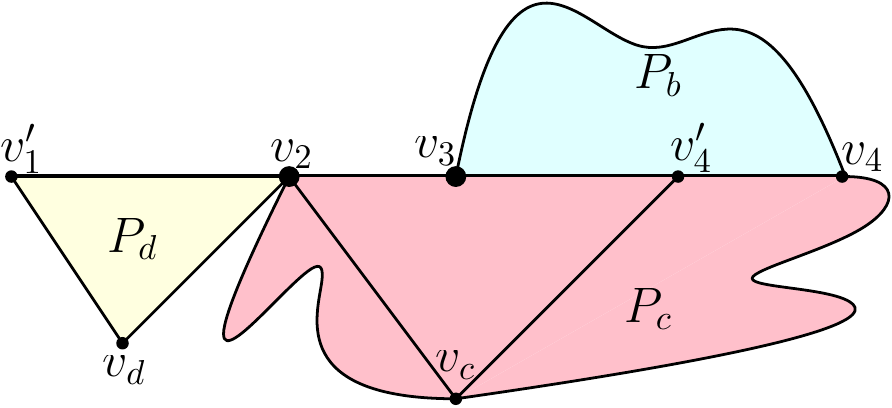}
\caption{$v_2$ and $v_3$ are vertices of $P$. $|P_b| = 3k_b+1$, $|P_c| = 3k_c$ and $|P_d| = 3k_d$. }
\label{fig:PointKernel3}
\vspace{-10pt}
\end{wrapfigure}
We are still in case $2.1$ and all its assumptions are still applicable. If the above scenario (about dissecting $P_c$ via a line that contains $v_3$) has not worked, we consider the case that we tried to avoid all along: the dissection of $P'$ along $\overline{v_1 v_4}$ such that the successive partitioning of $P_2$ created a subpolygon with a single-point kernel at $v_2$ whose edge contains $v_3$ in its interior. In other words, we assume that $P_2 = P_c \cup P_d$ has undergone some partitioning that resulted in the creation of the pentagon $v_1' v_4' v_c v_2 v_d$ for $v_c, v_4' \in P_c$, $v_d, v_1' \in P_d$ such that $v_1', v_4' \in L(v_2, v_3)$ and $v_4' \in \overline{v_3 v_4}$ (refer to Fig.~\ref{fig:PointKernel3}). It is impossible to localize an agent in the pentagon $v_1' v_4' v_c v_2 v_d$ with a pair of towers only. However, if $v_4' \in \overline{v_2 v_3}$ or $v_4' = v_3$ then a pair of towers in the vicinity of $v_2$ can oversee the pentagon, because $v_3$ is not causing an obstruction in this case. Despite that $P_a$ is a line segment, the original polygon $P$ has a non-empty interior in the vicinity of $v_2$ to the right of the ray $\overrightarrow{v_3 v_2}$. Notice that $v_1'$ may be equal to $v_1$; $v_4'$ may be equal to $v_4$. 
Moreover, $v_c$ and $v_4'$ may not be vertices of $P_1$ but were created during the partition. We also assume that $P_1 = P_b \cup \{v_1\}$ was not partitioned yet. Let $P_c'$ be a subpolygon of $P_c$ that inherits the partition of $P_c$ to the side of the dissection $\overline{v_c v_4'}$ that contains $v_4$ (refer Fig.~\ref{fig:PointKernel45}). Notice that $v_c$ and $v_3$ can see each other.  If $\overline{v_c v_4'}$ is a diagonal dissection then instead use $\overline{v_c v_3}$. In this case $v_4'$ can be deleted and $v_3$ is added to $P_c'$. The size of $P_c'$ does not change and the number of guards required to guard $P_c'$ does not increase. A pair of towers in the vicinity of $v_2$ can oversee the pentagon $v_1' v_3 v_c v_2 v_d$.

If the dissection $\overline{v_c v_4'}$ is a continuation of an edge of $P_c$ (let $v$ be a vertex that is adjacent to this edge) then two cases are possible. In the first case, the edge $\overline{v_c v}$ that produced the dissection is contained in $\overline{v_c v_4'}$ (refer Fig.~\ref{fig:PointKernel4}). In this case, instead of dissecting along $\overline{v_c v_4'}$ use $\overline{v v_3}$. The size of $P_c'$ does not change and the hexagon $v_1' v_3 v v_c v_2 v_d$ can be guarded by a pair of towers in the vicinity of $v_2$ (because $v_2 v_c v v_3$ is convex). In the second case the edge $\overline{v_c v}$ that produced the dissection $\overline{v_c v_4'}$ is not contained in $\overline{v_c v_4'}$ (refer Fig.~\ref{fig:PointKernel5}). In the worst case the size of $P_c'$ is $3k_{c'}+2$. If instead of dissecting along $\overline{v_c v_4'}$ we use $\overline{v_c v_3}$ then the size of $P_c'$ increases by $1$, which results in an increased number of guards. But, if we dissect along $\overline{v_c v_3}$ then the line segment $\overline{v_1 v_3}$ can be removed from $P_1$ (because it is guarded by the guard of $v_1' v_3 v_c v_2 v_d$) and $P_b$ can be joined with the updated $P_c'$. The size of $(P_b \cup P_c' \cup\{v_c\}) \setminus\{v_4'\}$ is $3(k_b+k_{c'})+2$, so the number of guards of $P'$ does not increase as a result of adjusting the partition. However, the current partition of $P_c'$ may no longer be relevant and may require repartition as part of the polygon $(P_b \cup P_c' \cup\{v_c\}) \setminus\{v_4'\}$. 

\begin{figure}
\centering
\subfigure[]{%
		\label{fig:PointKernel4}%
		\includegraphics[width=0.34\textwidth]{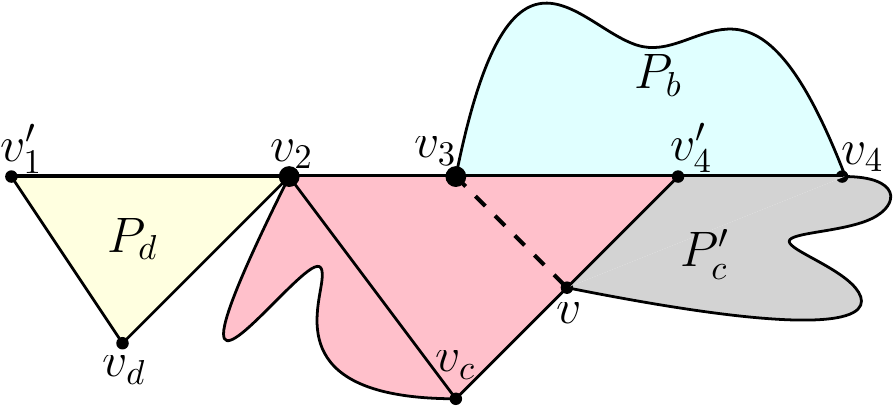}}
\hspace{0.1\textwidth}
\subfigure[]{%
		\label{fig:PointKernel5}%
		\includegraphics[width=0.34\textwidth]{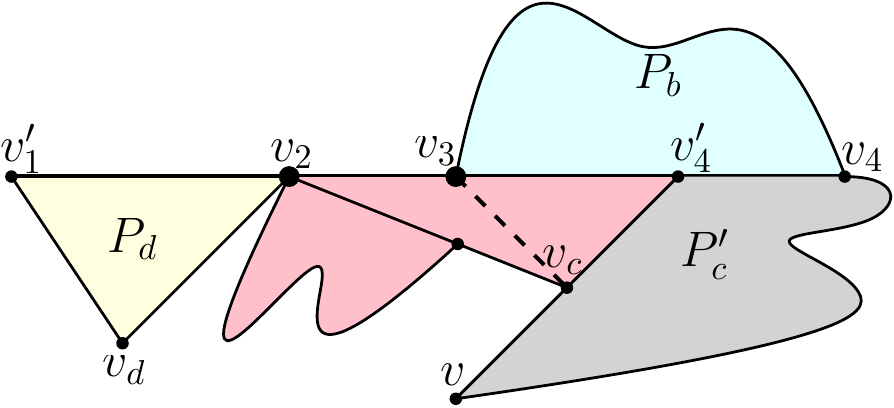}}
\caption{We avoid dissecting along $\overline{v_c v_4'}$.}
\label{fig:PointKernel45}
\end{figure}

{\color{blue} \textbf{Case 2.2:}} $|P_b| = 3k_b$. Then $|P_a| = 3k_a$. If $|P_c| = 3k_c$ and thus $|P_d| = 3k_d$ then $\overline{v_2 v_3}$ is a good diagonal dissection, which is a contradiction. Thus either $|P_c| = 3k_c+1$ and $|P_d| = 2$ or $|P_d| = 3k_d+1$ and $|P_c| = 2$. Those cases are symmetrical to the the case 2.1.

\subsection{$n = 3k+2$ and $P'$ has no good diagonal dissection}
\label{subsec:Partition}
\pdfbookmark[2]{$n = 3k+2$ and $P'$ has no good diagonal dissection}{subsec:Partition}

In this subsection we assume that $n = 3k+2$ and every diagonal of $P'$ decomposes it into polygons of size $n_1 = 3k_1$ and $n_2 = 3k_2+1$, where $k_1 + k_2 = k + 1$. Let $\mathcal{T}$ be a fixed triangulation of $P'$, and let $G(\mathcal{T})$ be the dual graph of $\mathcal{T}$. Notice that $G(\mathcal{T})$ has $n-2$ nodes. (Later, during the algorithm, we may change this triangulation.)

The proofs of the following Propositions~\ref{TothProp4} and~\ref{TothProp5} are identical to those in~\cite{Toth2000121}. We prove the following Lemmas~\ref{TothLem1},~\ref{TothLem2} and~\ref{TothLem3} in this paper.

\begin{mydefP}
\label{TothProp4}
$G(\mathcal{T})$ has exactly $k+1$ leaves.
\end{mydefP}

\begin{mydefL}
\label{TothLem1}
If $P'$ has $n = 3k+2$ vertices and has no good dissection then $P'$ has at most $k$ reflex angles.
\end{mydefL}

\begin{mydefP}
\label{TothProp5}
If $P'$ has at most $k$ reflex angles, then $k$ $180^\circ$-guards can monitor $P'$.
\end{mydefP}

\begin{mydefL}
\label{TothLem2}
If $P'$ has size $n = 3k+2$ then at least one of the following two statements is true:
\begin{enumerate}
\item $P'$ has a good dissection.
\item For every triangle $\triangle A B C$ in $\mathcal{T}$ that corresponds to a leaf of $G(\mathcal{T})$ with $AC$ being a diagonal of $P'$, either $\angle A < \pi$ or $\angle C < \pi$ in $P'$.
\end{enumerate}
\end{mydefL}

\begin{mydefL}
\label{TothLem3}
If a convex angle of $P'$ is associated to two leaves in $G(\mathcal{T})$, then $\left \lfloor \frac{n}{3} \right \rfloor$ $180^\circ$-guards can monitor~$P'$. 
\end{mydefL}

\subsubsection{Adaptation of T{\'o}th's Lemma~\ref{TothLem3} to our problem} 
\label{subsubsec:Lemma3}
\pdfbookmark[3]{Adaptation of T{\'o}th's Lemma~\ref{TothLem3} to our problem.}{subsubsec:Lemma3}

\hfill\\
\\
Let $v_1$ be a vertex at a convex angle of $P'$ associated to two leaves in $G(\mathcal{T})$ (refer to Fig.~\ref{fig:TothLemma3}). Let $P_1$ be the polygon formed by all triangles of $\mathcal{T}$ adjacent to $v_1$. Notice that $P_1$ is a fan; $v_1$ is a center of the fan; and the dual of the triangulation of $P_1$,  inherited from $P'$, is a path of nodes. Recall that the center of a fan $P_f$ is a vertex of $P_f$ that can see all other vertices of $P_f$.

\begin{wrapfigure}{r}{0.6\textwidth}
\centering
\includegraphics[width=0.58\textwidth]{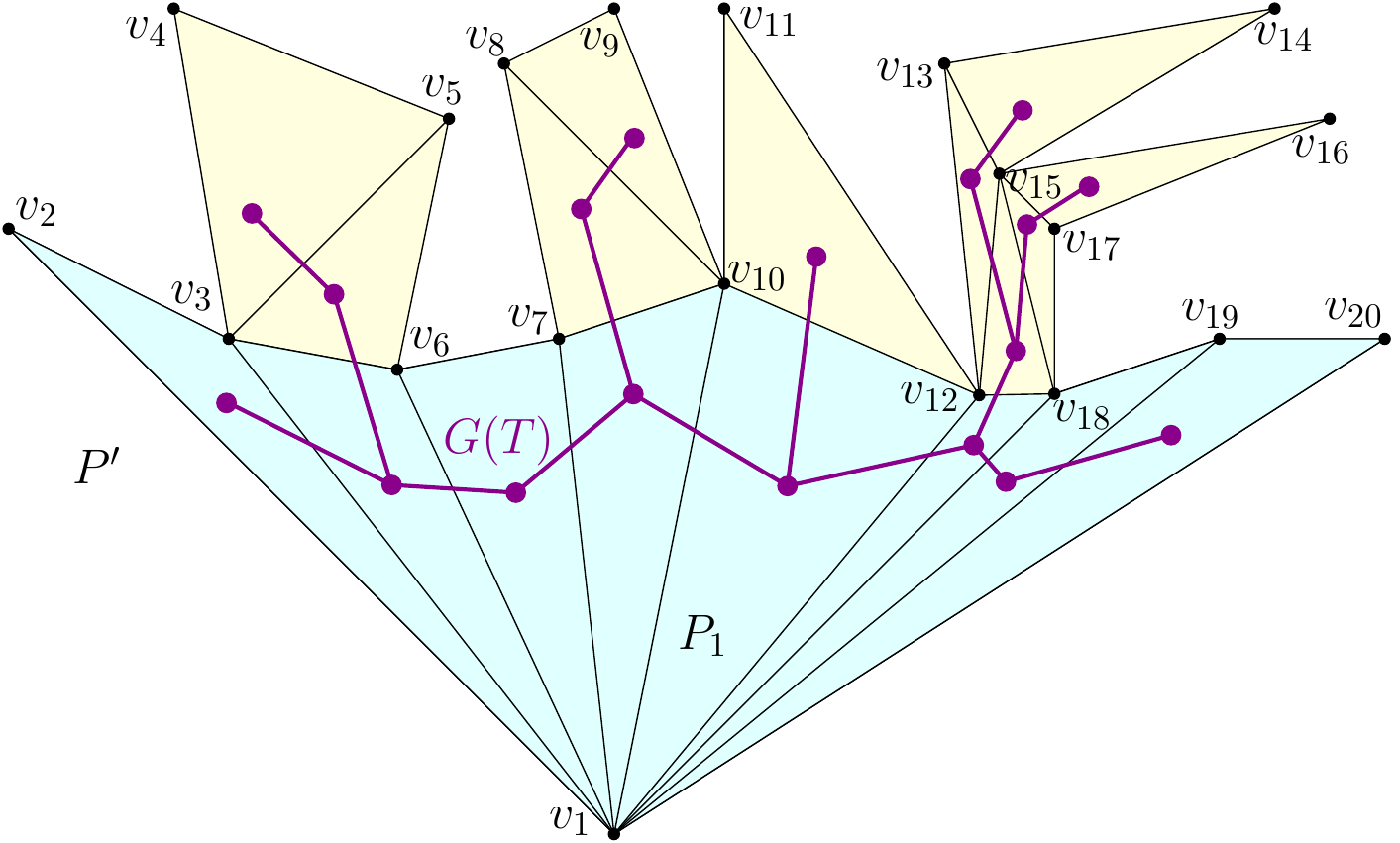}
\caption{$P'$ has size $n = 20 = 3k + 2$ for $k = 6$, and does not have good diagonal dissection. $G(\mathcal{T})$ is drawn on top of $\mathcal{T}$. $P_1$ is a fan with dominant point $v_1$ and it is highlighted in cyan.}
\label{fig:TothLemma3}
\vspace{-30pt}
\end{wrapfigure}

The original proof of T{\'o}th's Lemma~\ref{TothLem3} still holds for polygons*. 
However, this is one of the places where a $180^\circ$-guard is explicitly positioned at $v_1$ to monitor all the triangles of $\mathcal{T}$ adjacent to $v_1$. If $kernel(P_1)$ is not a single point, then the positioning can be reused for our problem. We can put a pair of towers on the same edge of $P'$ in $kernel (P_1) \cap \partial P_1$ to locate an agent in $P_1$. However, since there can be degenerate triangles among those adjacent to $v_1$, the kernel of $P_1$ can degenerate into a single point: $kernel (P_1) = v_1$. In this case, the location of an agent in $P_1$ cannot be determined with a pair of towers only. To overcome this problem we prove the following lemma.

\begin{lemma}
\label{lem:TL3}
There exist a triangulation of $P'$ such that its part, inherited by $P_1$ (i.e. triangles incident to $v_1$), does not contain any degenerate triangle.
\end{lemma}
\begin{proof} 
Let us first observe that a triangle of $\mathcal{T}$, corresponding to a leaf in $G(\mathcal{T})$, cannot be degenerate due to the simplification step performed on $P'$. There are two cases to consider:

{\color{red} \textbf{Case 1:}} Suppose that the triangulation/fan of $P_1$ contains two or more degenerate triangles adjacent to each other. Let $v_1, u_1, u_2, \ldots, u_i$ for some $i > 2$ be the vertices of the degenerate triangles sorted according to their distance from $v_1$. Notice that $v_1, u_1, u_2, \ldots, u_i$ belong to the same line. Those degenerate triangles must be enclosed in $P_1$ between a pair of non-degenerate triangles, let us call them $\triangle_1$ and $\triangle_2$. Since $P'$ does not have angles of size $2\pi$, the diagonal $\overline{v_1 u_1}$ must be shared with one of $\triangle_1$ or $\triangle_2$. Assume, without loss of generality, that $u_1$ is a vertex of $\triangle_1$. Let $u_j$ be a vertex of $\triangle_2$, for some $1 < j \leq i$. It is possible to re-triangulate $P'$ such that $P_1$ will contain only one degenerate triangle $\triangle v_1 u_1 u_j$ between $\triangle_1$ and $\triangle_2$ and other triangles of $P_1$ (that are not between $\triangle_1$ and $\triangle_2$) will not be affected. The shape of $P_1$ may or may not change but its size will decrease.

{\color{red} \textbf{Case 2:}} Suppose that the triangulation/fan of $P_1$ contains one degenerate triangle $\triangle v_1 u_1 u_j$ enclosed between a pair of non-degenerate triangles $\triangle_1$ and $\triangle_2$. Assume, without loss of generality, that $\triangle_2$ shares a diagonal $\overline{v_1 u_j}$ with $\triangle v_1 u_1 u_j$. Let $w$ be the third vertex of $\triangle_2$, so $\triangle_2 = \triangle v_1 u_j w$. Notice that $u_1$ and $w$ see each other, because $u_1$ belongs to the line segment $\overline{v_1 u_j}$. We can flip the diagonal $\overline{v_1 u_j}$ into $\overline{u_1 w}$ in $\mathcal{T}$. As a result, $P_1$ will now contain $\triangle v_1 u_1 w$ instead of $\triangle v_1 u_1 u_j$ and $\triangle_2$. Notice that $\triangle v_1 u_1 w$ is non-degenerate. 

We showed that we can obtain a triangulation of $P'$ in which all the triangles incident to $v_1$ are non-degenerate. Thus, $P_1$ will have no degenerate triangles and still contain vertex $v_1$ together with two leaves of $G(\mathcal{T})$ associated with $v_1$.
\qed
\end{proof}

\subsubsection{Proof of T{\'o}th's Lemma~\ref{TothLem2} and its adaptation to our problem}
\label{subsubsec:Lemma2}
\pdfbookmark[3]{Proof of T{\'o}th's Lemma~\ref{TothLem2} and its adaptation to our problem.}{subsubsec:Lemma2}
\hfill\\
\\
Section~\ref{subsubsec:Lemma2} is devoted to the proof of T{\'o}th's Lemma~\ref{TothLem2}.

T{\'o}th defines two types of elements of $G(\mathcal{T})$. A leaf of $G(\mathcal{T})$ is called a \emph{short leaf} if it is adjacent to a node of degree $3$. If 
a leaf of $G(\mathcal{T})$ is adjacent to a node of degree $2$ then this leaf is called a \emph{long leaf}. Since we are under the assumption that $n = 3k+2$ and $P'$ has no good diagonal dissection then the node of $G(\mathcal{T})$ adjacent to a long leaf is also adjacent to a node of degree~$3$.

\begin{wrapfigure}{r}{0.32\textwidth}
\centering
\includegraphics[width=0.3\textwidth]{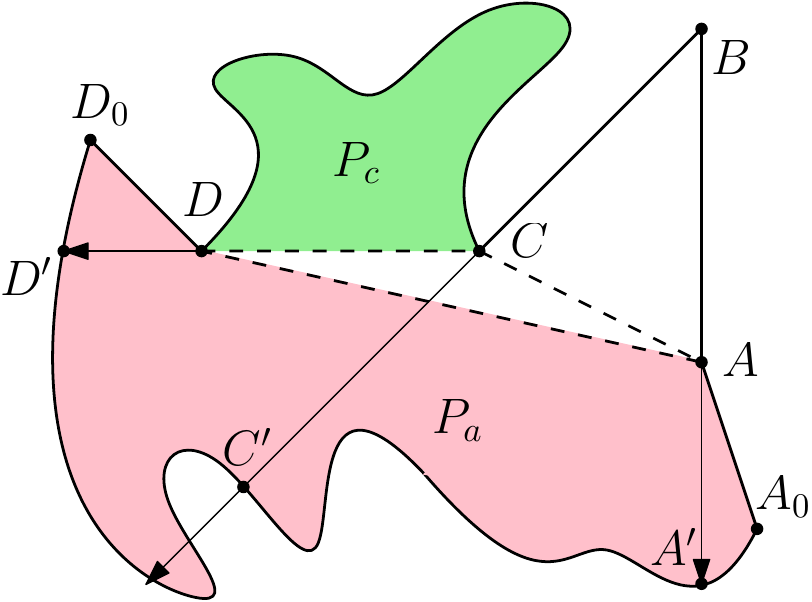}
\caption{$\triangle ABC$ corresponds to a short leaf in $G(\mathcal{T})$. $\triangle ACD$ can be degenerate.}
\label{fig:TothShortLeaf_1}
\vspace{-20pt}
\end{wrapfigure}
In this subsection we keep T{\'o}th's original names and notations to simplify cross-reading.

The angle $\angle ABC$ is the angle that the ray $\overrightarrow{BA}$ makes while rotating counter-clockwise towards the ray $\overrightarrow{BC}$. For example, the angle of $P'$ at $B$ is $\angle CBA$ (refer to Fig.~\ref{fig:TothShortLeaf_1}).

Let $\triangle ABC$ correspond to a short leaf in $G(\mathcal{T})$, where $\overline{AC}$ is a diagonal of $P'$, and let $\triangle ACD$ correspond to a node in $G(\mathcal{T})$ adjacent to the leaf $\triangle ABC$. Refer to Fig.~\ref{fig:TothShortLeaf_1}. Notice that $\triangle ABC$ cannot be degenerate, otherwise the vertex $B$ would be deleted during the simplification step. However, $\triangle ACD$ can be degenerate. The diagonals $\overline{AD}$ and $\overline{CD}$ decompose $P'$ into polygons $P_a$, $ABCD$ and $P_c$, where $A \in P_a$ and $C \in P_c$. Let $n_a$ be size of $P_a$ and $n_c$ be size of $P_c$.

\begin{mydefC}
\label{claimT1}
$n_a = 3 k_a + 1$ and $n_c = 3 k_c + 1$.
\end{mydefC}

\begin{mydefC}
\label{claimT2}
$ABCD$ is a non-convex quadrilateral, i.e. it has a reflex vertex at $A$ or $C$.
\end{mydefC}

The original proof of both claims (refer to~\cite{Toth2000121}) can be reused for polygons*, 
so we do not present their proofs here. However, T{\'o}th's Claim~\ref{claimT3} (presented below) requires a slightly different proof to hold true for polygons*. In addition, there was a mistake in the original proof of T{\'o}th's Claim~\ref{claimT3} that we fixed here. But first we would like to clarify the framework we work in, to explain tools and assumptions we use.

Assume, without loss of generality, that the reflex vertex of $ABCD$ is at $C$ (which exists by T{\'o}th's Claim~\ref{claimT2}). Assume that $\mathcal{T}$ is the triangulation of $P'$ in which $n_a$ is \textbf{minimal}. This assumption together with the fact that $P'$ does not have a good diagonal dissection implies that there does not exist a vertex of $P_a$ in the interior of a line segment $\overline{DA}$. By T{\'o}th's Claim~\ref{claimT1}, $n_a \geq 4$ (notice that $n_a \neq 1$ because $A$, $D \in P_a$). Let $A_0$ and $D_0$ be vertices of $P_a$ adjacent to $A$ and $D$ respectively (refer to Fig.~\ref{fig:TothShortLeaf_1}). By $\overrightarrow{uv}$ we denote the ray that starts at $u$ and passes through $v$. Let $A'$ be a point such that $A' \in \partial P' \cap \overrightarrow{BA}$ and there exists a point $A'' \in \partial P'$ strictly to the \textbf{right} of $\overrightarrow{BA}$ such that $A'$ and $A''$ belong to the same edge of $P'$ and both visible to $B$ and $A$. Let $D'$ be a point such that $D' \in \partial P' \cap \overrightarrow{CD}$ and there exist a point $D'' \in \partial P'$ strictly to the \textbf{left} of $\overrightarrow{CD}$ such that $D'$ and $D''$ belong to the same edge of $P'$ and both visible to $C$ and $D$. Let $C'$ be a point defined similarly to $D'$ but with respect to the ray $\overrightarrow{BC}$.
Notice that if $A=A'$ then $\angle A = \angle B A A_0 < \pi$ in $P'$ and thus the second condition of T{\'o}th's Lemma~\ref{TothLem2} holds. Therefore, assume that $A \neq A'$.

\begin{mydefC}
\label{claimT3}
The points $A'$ and $D'$ belong to the same edge of $P'$.
\end{mydefC}
\begin{proof}
Let us rotate the ray $\overrightarrow{CA'}$ around $C$ in the direction of $D'$. Notice that in the original proof, T{\'o}th uses $\overrightarrow{CA}$. However, there are polygons (even in general position) for which T{\'o}th's proof does not hold. For example, when the ray $\overrightarrow{CA}$ hits $A_0$, then T{\'o}th claims that $\angle BAA_0 < \pi$. Figure~\ref{fig:TothShortLeaf_1} can serve as a counterexample to this claim, because $A_0$ is indeed the first point hit by the rotating ray $\overrightarrow{CA}$, however $\angle BAA_0 > \pi$.

The structure of the original proof can be used with respect to $\overrightarrow{CA'}$, assuming that $A'$ is visible to $C$.
Thus, assume first that $C$ and $A'$ can see each other. We rotate $\overrightarrow{CA'}$ around $C$ in the direction of $D'$. Let $O$ be the first vertex of $P'$ visible from $C$ that was hit by the ray (in case there are several collinear such vertices of $P_a$, then let $O$ be the one that is closest to $C$). 

If $O = D$ then $A'$ and $D'$ belong to the same edge of $P'$ (notice that it is possible that $D' = D$), so the claim holds.

If $O = A_0$ then $A=A'$, $\angle A = \angle B A A_0 < \pi$ and thus the second claim in T{\'o}th's Lemma~\ref{TothLem2} holds.

If $O \neq D$ and $O \neq A_0$, then $\overline{AO}$ and $\overline{CO}$ are diagonals of $P'$. Refer to Fig.~\ref{fig:TothShortLeaf_2}. There exists a triangulation of $P'$ that contains $\triangle A C O$ and has $\triangle A B C$ as a short leaf. Consider quadrilateral $ABCO$. It is non-convex by T{\'o}th's Claim~\ref{claimT2}. By construction, $O$ is to the right of $\overrightarrow{BA}$, thus $\angle B A O < \pi$. It follows, that the only possible reflex angle in $ABCO$ is $\angle O C B$, which is a contradiction to the minimality of $P_a$ (recall, we assumed that we consider the triangulation of $P'$ in which $n_a$ is minimal), and thus, such a vertex $O$ does not exist.

\begin{figure}
\centerline{\resizebox{!}{0.25\textwidth}{\includegraphics{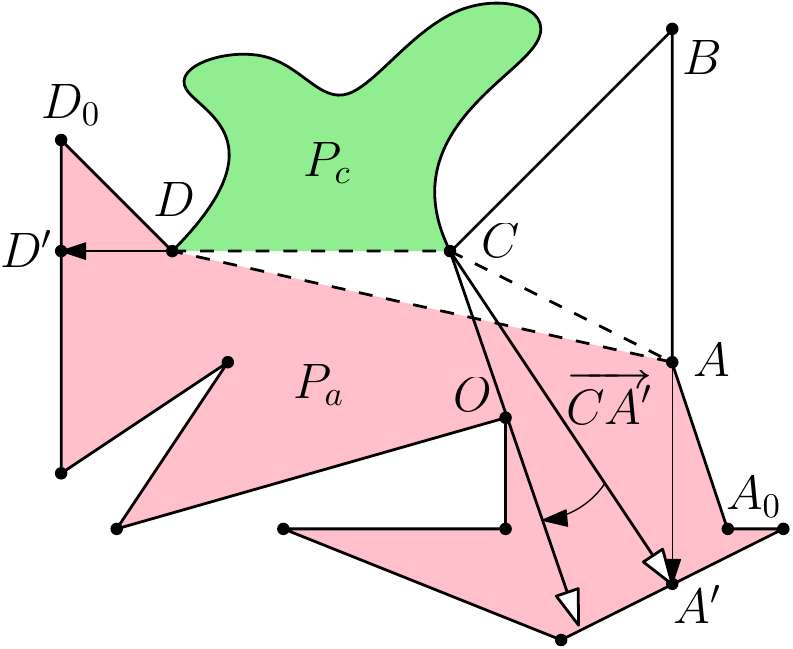}}}
\caption{$\triangle ABC$ corresponds to a short leaf in $G(\mathcal{T})$; $\triangle ACD$ can be degenerate. $\angle D_0 D C > \pi$.}
\label{fig:TothShortLeaf_2}
\end{figure}

Assume now that $A'$ is \textbf{not} visible to $C$. It follows that some part of $\partial P'$ belongs to the interior of $\triangle A A'C$. Recall, that $A'$ is defined as a point such that $A' \in \partial P' \cap \overrightarrow{BA}$ and there exists a point $A'' \in \partial P'$ strictly to the \textbf{right} of $\overrightarrow{BA}$ such that $A'$ and $A''$ belong to the same edge of $P'$ and both visible to $B$ and $A$. Thus, there is no intersection between the part of $\partial P'$ that belongs to the interior of $\triangle A A'C$ and $\overline{AA'}$. It follows, that $\partial P'$ crossed $\overline{CA'}$ at least twice, and thus there must be a vertex of $P_a$ interior to $\triangle A A'C$. Among all the vertices of $P_a$ that lie in $\triangle A A'C$, let $O'$ be the vertex that is closest to the line passing through $\overline{CA}$. Notice, that $O'$ is visible to $B$ (because $\triangle B A' C$ contains $\triangle A A'C$). Notice also that $\overline{AO'}$ and $\overline{CO'}$ are diagonals of $P'$. One of $\overline{AO'}$, $\overline{BO'}$ or $\overline{CO'}$ is a good diagonal dissection, which is a contradiction to the main assumption of Sect~\ref{subsec:Partition} that $P'$ has no good diagonal dissection. Thus, $A'$ is visible to $C$.
\qed
\end{proof}


It follows from T{\'o}th's Claim~\ref{claimT3} that the quadrilateral $A'BCD'$ has no common points with $\partial P'$ in its interior, but on the boundary only.

We have derived several properties satisfied by $P'$ and now we are ready to show the existence of good (non-diagonal) dissections that consist of one or two line segments. We discuss the following three cases, that span over Claims $4,5$ in~\cite{Toth2000121}.

{\color{red} \textbf{Case 1:}} $\angle D_0 D C < \pi$. In this case $D = D'$. It follows from T{\'o}th's Claim~\ref{claimT3} that $A', C' \in \overline{D D_0}$. Refer to Fig.~\ref{fig:TothShortLeaf_3}.

Line segment $\overline{CC'}$ represents a good dissection that splits $P'$ into two polygons: $P_c \cup \triangle CC'D$ of size $n_c + 1 = 3 k_c + 2$ and $(P' \setminus P_c) \setminus \triangle CC'D$ of size $n_a + 1 = 3 k_a + 2$. Notice that $\overline{BC'}$ is an edge of the subpolygon of $P'$ to the left of $\overrightarrow{BC}$ and thus $C$ is not a vertex of this subpolygon. If $\triangle ACD$ is degenerate (in which case $C$ is between $A$ and $D$), then $\overline{CC'}$ is still a good dissection. However, if $D_0$ is also collinear with $D$, $C$ and $A$, then $D_0$ is visible to $C$ and $\overline{CD_0}$ represents a good diagonal dissection. 

\begin{figure}
\centerline{\resizebox{!}{0.24\textwidth}{\includegraphics{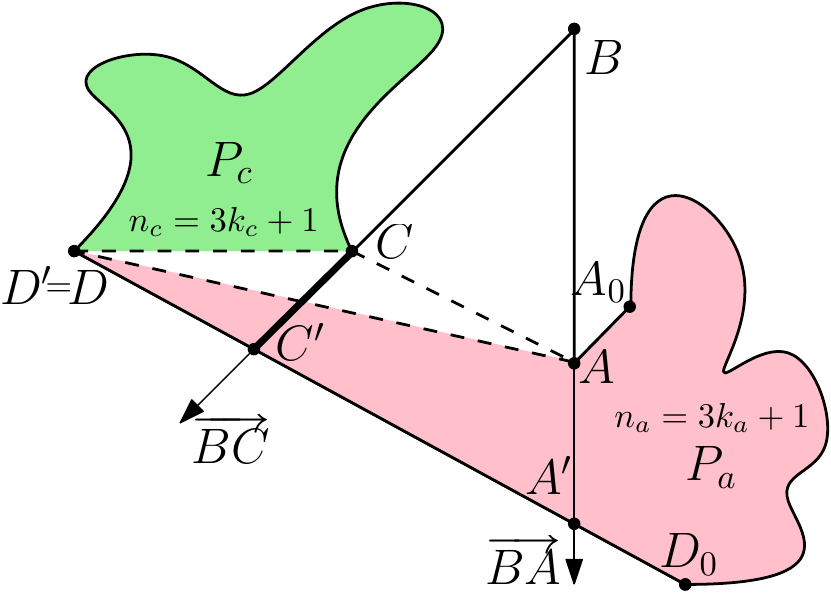}}}
\caption{$\triangle ABC$ corresponds to a short leaf in $G(\mathcal{T})$; $\triangle ACD$ can be degenerate. $\angle D_0 D C < \pi$; $\overline{CC'}$ is a good dissection.}
\label{fig:TothShortLeaf_3}
\end{figure}

{\color{red} \textbf{Case 2:}} $\angle D_0 D C = \pi$. In this case $D_0 = D'$ and $\overline{CD}$ is a good diagonal dissection that splits $P'$ into $P_c$ of size $n_c = 3 k_c + 1$ and $P_a \cup ABCD$ of size $n_a + 1= 3 k_a + 2$. Notice that $\overline{CD_0}$ is an edge in $P_a \cup ABCD$ and thus $D$ is not a vertex in $P_a \cup ABCD$.

{\color{red} \textbf{Case 3:}} $\angle D_0 D C > \pi$.  Let $D_0'$ be the point closest to $D$ where the ray $\overrightarrow{D_0 D}$ reaches $\partial P'$. If the line segments $\overline{CC'}$ and $\overline{DD_0'}$ intersect inside the quadrilateral $C A A' D'$ at $Q$ (refer to Fig.~\ref{fig:TothShortLeaf_4}), then $\overline{DQ} \cup \overline{QC}$ is a good dissection, that splits $P'$ into polygon $P_c \cup \triangle CDQ$ of size $n_c + 1 = 3 k_c + 2$ and polygon $P_a \cup DQBA$ of size $n_a + 1 = 3 k_a + 2$. 
However, if $C A A' D'$ degenerates into a line segment (which happens when $\triangle ACD$ is degenerate and there exists an edge $\overline{IJ}$ of $P_a$ that contains $\overline{AD}$) then $Q$ cannot be defined. Refer to Fig.~\ref{fig:TothShortLeaf_5}. 
\begin{figure}
\centering
\subfigure[]{%
		\label{fig:TothShortLeaf_4}%
		\includegraphics[width=0.35\textwidth]{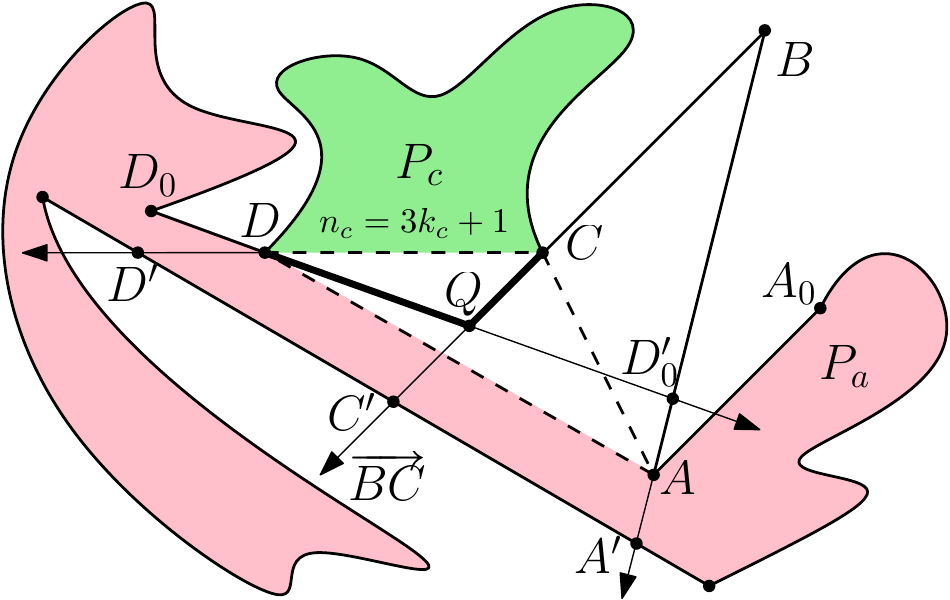}}
\hspace{0.1\textwidth}
\subfigure[]{%
		\label{fig:TothShortLeaf_5}%
		\includegraphics[width=0.36\textwidth]{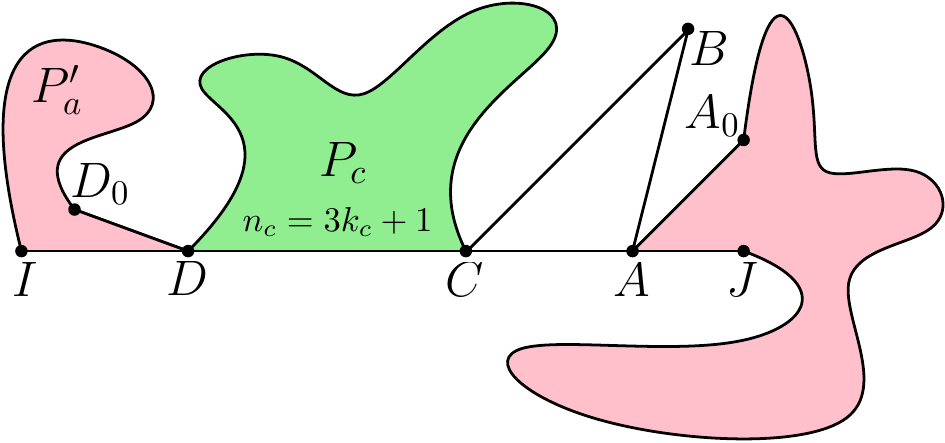}}
\caption{$\triangle ABC$ corresponds to a short leaf in $G(\mathcal{T})$; $\angle D_0 D C > \pi$. \textbf{(a)} $\overline{DQ} \cup \overline{QC}$ is a good dissection of $P'$. \textbf{(b)} $\triangle ACD$ is degenerate; the edge $\overline{IJ}$ of $P_a$ contains $\overline{DA}$. $P_a$ is highlighted in pink. Notice that $C$ is not a vertex of $P_a$. $P_a'$ is a subpolygon of $P_a$ that contains $D_0$ and $\overline{ID}$ is its edge.}
\label{fig:TothShortLeaf_45}
\end{figure}
In this case we show that $P'$ has a good diagonal dissection. Notice that $\overline{ID}$ is a diagonal of $P_a$; it splits $P_a$ into two subpolygons. Let $P_a'$ be a subpolygon of $P_a$ that contains $D_0$. Let $n_a'$ be the size of $P_a'$. We consider three cases:
\begin{itemize}
\item[$\blacktriangleright$] $n_a' = 3k_a'$: In this case $\overline{IA}$ is a good diagonal dissection. The size of $P_a' \cup P_c \cup \triangle ABC$ is $3k_a' + 3k_c + 1 + 3 - 2 = 3(k_a' + k_c) +2$. Notice that the ``$-2$'' in the previous formula stands for vertices $D$ and $C$ that were counted twice. 
\item[$\blacktriangleright$] $n_a' = 3k_a' + 1$: In this case $\overline{CJ}$ is a good diagonal dissection. The size of $P_a' \cup P_c \cup IDCJ$ is $3k_a' + 1 + 3k_c + 1 + 1 - 1 = 3(k_a' + k_c) +2$.
\item[$\blacktriangleright$] $n_a' = 3k_a' + 2$: In this case $\overline{ID}$ is a good diagonal dissection.
\end{itemize}

In this subsection we assumed that $P'$ has no good diagonal dissection, thus we deduce that $C A A' D'$ cannot degenerate into a line segment. Notice that $Q$ exists even when $\triangle ACD$ is degenerate.

\begin{wrapfigure}{r}{0.42\textwidth}
\centering
\includegraphics[width=0.36\textwidth]{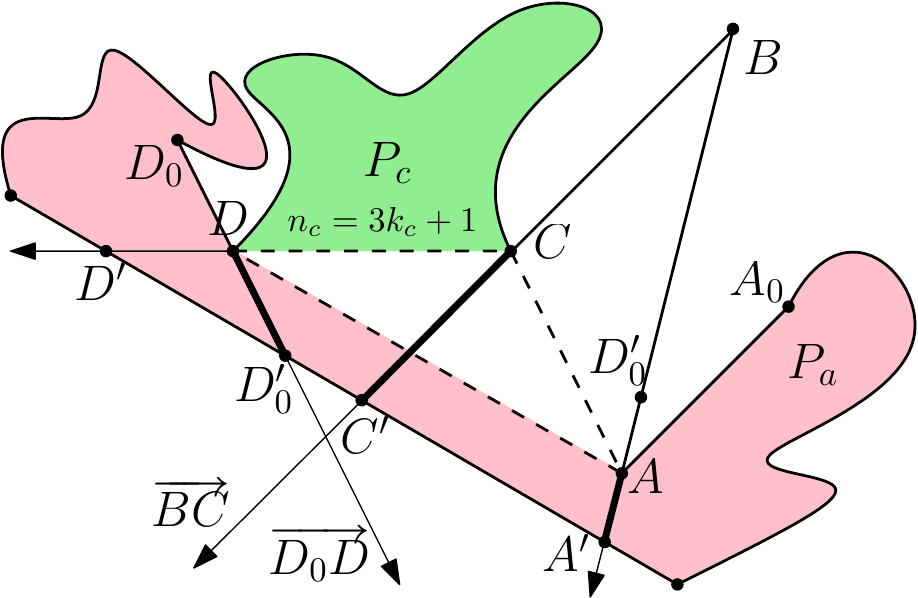}
\caption{$\triangle ABC$ corresponds to a short leaf in $G(\mathcal{T})$; $\angle D_0 D C > \pi$; $\triangle ACD$ can be degenerate. One of $\overline{DD_0'}$, $\overline{CC'}$ and $\overline{AA'}$ is a good dissection of $P'$.} 
\label{fig:TothShortLeaf_6}
\vspace{-20pt}
\end{wrapfigure}
If $\overline{CC'}$ and $\overline{DD_0'}$ do not intersect inside $C A A' D'$, then $D_0'$ belongs to the line segment $C'D'$. Refer to Fig.~\ref{fig:TothShortLeaf_6}. T{\'o}th shows in Claim $5$ in~\cite{Toth2000121} that one of the line segments  $\overline{DD_0'}$, $\overline{CC'}$ and $\overline{AA'}$ is a good dissection.

It was shown so far that if $\triangle ABC$ is a short leaf in $G(\mathcal{T})$, then either $P'$ has a good dissection or the angle at vertex $A$ or $C$ in $P'$ is convex. It is left to prove that T{\'o}th's Lemma~\ref{TothLem2} is true for long leaves. Notice that if $\triangle ABC$ is a long leaf in $G(\mathcal{T})$ but there exists a triangulation of $P'$ where $\triangle ABC$ is a short leaf then T{\'o}th's Lemma~\ref{TothLem2} is true for $\triangle ABC$.

\medskip
Let $\triangle ABC$ be a long leaf of $G(\mathcal{T})$ such that there does not exist a triangulation of $P'$ where $\triangle ABC$ is a short leaf. Recall that in this subsection, we assumed that $P'$ has no good diagonal dissection. Thus, the node of $G(\mathcal{T})$ adjacent to a long leaf is also adjacent to a node of degree $3$. We also concluded that $\triangle ABC$ cannot be degenerate.

\stepcounter{mydefC}
\stepcounter{mydefC}
\begin{mydefC}
\label{claimT6}
If $\triangle ABC$ is a long leaf of $G(\mathcal{T})$ for every triangulation $\mathcal{T}$ of $P'$ then the node of $G(\mathcal{T})$ adjacent to the node $\triangle ABC$ corresponds to the same triangle for every~$\mathcal{T}$.
\end{mydefC}

T{\'o}th's Claim~\ref{claimT6} is true for a triangulation $\mathcal{T}$ that contains degenerate triangles and thus it is true for the $P'$ defined in this paper. 

Let $\triangle ACD$ be a triangle adjacent to $\triangle ABC$ in $\mathcal{T}$. The ray $\overrightarrow{CA}$ (respectively $\overrightarrow{CD}$, $\overrightarrow{BC}$, $\overrightarrow{BA}$) reaches $\partial P'$ at $A'$ (respectively $D'$, $C'$, $B'$). Notice that $A'$ is defined differently than in the case with short leaves. Refer to Fig.~\ref{fig:TothShortLeaf_7}.
By T{\'o}th's Claim~\ref{claimT6}, $\triangle ACD$ is unique. Notice that $\triangle ACD$ can be degenerate, but because it is unique, it does not contain any other vertex of $P'$. Moreover, there does not exist an edge $\overline{IJ}$ of $P'$ that contains $\triangle ACD$, otherwise $P'$ has a good diagonal dissection since $A$, $C$, $D$, $I$ and $J$ can see each other. T{\'o}th's Claim~\ref{claimT7}, that follows this discussion, is thus true for polygons whose triangulation may contain degenerate triangles.

\begin{figure}
\centering
\includegraphics[width=0.5\textwidth]{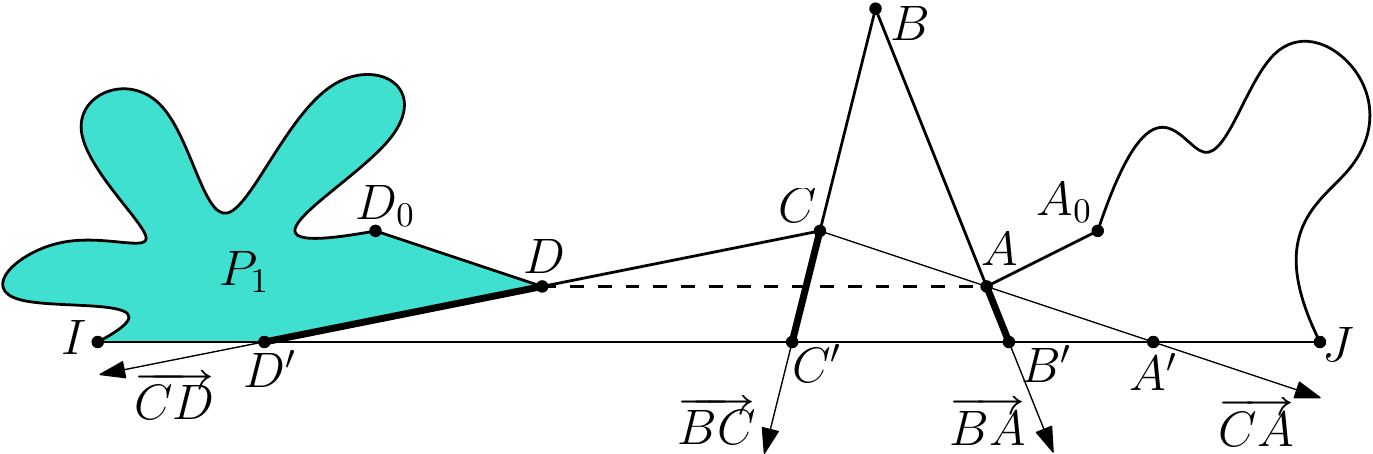}
\caption{$\triangle ABC$ corresponds to a long leaf in $G(\mathcal{T})$. $\overline{IJ}$ is an edge of $P'$; it contains $A'$, $B'$, $C'$ and $D'$. One of $\overline{DD'}$, $\overline{CC'}$ or $\overline{AB'}$ is a good dissection of $P'$. $P_1$ is highlighted in cyan.} 
\label{fig:TothShortLeaf_7}
\end{figure}

\begin{mydefC}
\label{claimT7}
The points $A'$ and $D'$ belong to the same edge of $P'$ or the angle of $P'$ at $A$ is convex and thus the second condition of T{\'o}th's Lemma~\ref{TothLem2} holds for $\triangle ABC$. Refer to Fig.~\ref{fig:TothShortLeaf_7}.
\end{mydefC}

It follows that $C'$ and $B'$ belong to the same edge as $A'$ and $D'$.

Assume that the angles of $P'$ at $A$ and $C$ are reflex, otherwise the second condition of T{\'o}th's Lemma~\ref{TothLem2} holds for $\triangle ABC$ and our proof is complete. 

Consider the angle of $P'$ at $D$. It can be either convex or reflex. Notice that non of the angles of $P'$ equals $\pi$ or $0$ because of the simplification step. We discuss the following two cases, that span over Claims $8,9$ in~\cite{Toth2000121}, and show that in either case $P'$ has a good dissection.

{\color{red} \textbf{Case 1:}} $\angle D_0 D C > \pi$. One of the line segments $\overline{DD'}$, $\overline{CC'}$ or $\overline{AB'}$ is a good dissection of $P'$. Refer to Fig.~\ref{fig:TothShortLeaf_7}. $\overline{DD'}$ partitions $P'$ into two subpolygons. Let $P_1$ be one of them that contains $D_0$ (it is highlighted in cyan on Fig.~\ref{fig:TothShortLeaf_7}). The size of $P_1$ is $n_1 = 3k_1 + q_1$; the size of $P_1 \cup CDD'C'$ is $n_1 + 1$; the size of $P_1 \cup CDD'C' \cup ACC'B'$ is $n_1 + 2$. If $q_1 = 2$ (respectively $q_1 = 1$, $q_1 = 0$) then $\overline{DD'}$ (respectively $\overline{CC'}$, $\overline{AB'}$) is a good dissection of $P'$.

{\color{red} \textbf{Case 2:}} $\angle D_0 D C < \pi$. It follows that $D'=D$, $A' \in \overline{DD_0}$ and $A' \neq D_0$, otherwise $\overline{AD_0}$ is a good diagonal dissection. Refer to Fig.~\ref{fig:TothShortLeaf_89}. Let $A_0'$ be the point where $\overrightarrow{A_0A}$ reaches $\partial P'$. T{\'o}th's Claim~\ref{claimT7} implies that $A_0' \in \overline{CD}$ or $A_0' \in \overline{DB'}$. If $A_0' \in \overline{DB'}$ (refer to Fig.~\ref{fig:TothShortLeaf_8}) then $\overline{AA_0'}$ is a good dissection of $P'$. It creates a pentagon $ABCDA_0'$ and a polygon $P' \setminus ABCDA_0'$ whose size is $n-3 = 3(k-1)-2$ (notice that $\overline{AA_0'}$ is an edge of $P' \setminus ABCDA_0'$ and thus $A \notin P' \setminus ABCDA_0'$). Notice that if $A_0' = D$ then $A_0$ can see $D$ and thus $\overline{A_0D}$ is a good diagonal dissection~of~$P'$. 

\begin{figure}
\centering
\subfigure[]{%
		\label{fig:TothShortLeaf_8}%
		\includegraphics[width=0.36\textwidth]{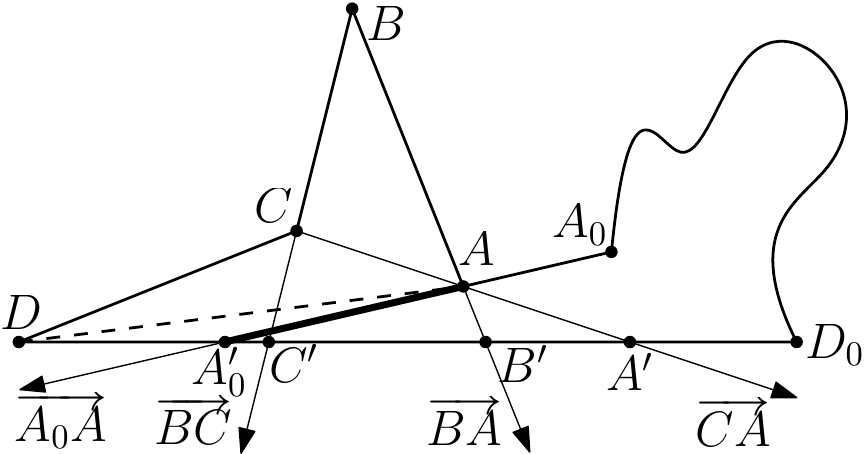}}
\hspace{0.1\textwidth}
\subfigure[]{%
		\label{fig:TothShortLeaf_9}%
		\includegraphics[width=0.4\textwidth]{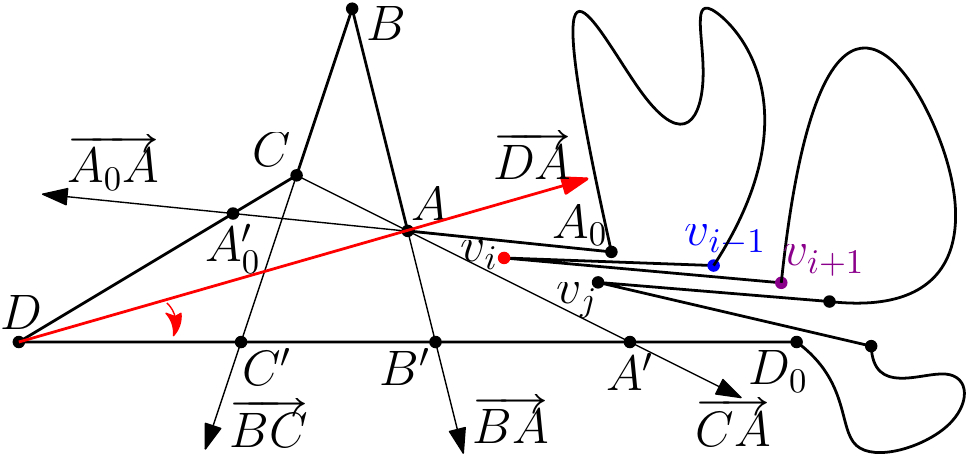}}
\caption{$\triangle ABC$ corresponds to a long leaf in $G(\mathcal{T})$; $\angle D_0 D C < \pi$. \textbf{(a)} Case where$A_0' \in \overline{DB'}$. $\overline{AA_0'}$ is a good dissection of $P'$. \textbf{(b)} Case where $A_0' \in \overline{CD}$.}
\label{fig:TothShortLeaf_89}
\end{figure}

Assume that $A_0' \in \overline{CD}$. Refer to Fig.~\ref{fig:TothShortLeaf_9}. Let us assign labels to the vertices of $P'$ according to their order around $\partial P'$ as follows: $v_0 = A$, $v_1 = A_0$, $v_2$, $\ldots$, $v_{n-4} = D_0$, $v_{n-3} = D$, $v_{n-2} = C$, $v_{n-1} = B$. $\overline{DA}$ is a diagonal of $P'$ and thus at least some interval of $\overline{AA_0}$ is visible to $D$. Let us rotate $\overrightarrow{DA}$ towards $A_0$. The ray hits $v_i$ for $1 < i \leq v_{n-5}$ (notice that the ray cannot hit $A_0$, otherwise  $\overline{DA_0}$ is a good diagonal dissection). Observe that the angle of $P'$ at $v_i$ must be reflex. Let $X_1$ (respectively $X_2$) be a point where $\overrightarrow{v_{i+1} v_i}$ (respectively $\overrightarrow{v_{i-1}v_i}$) reaches $\partial P'$. By construction, $X_1 \in \overline{AA_0}$ or $X_1 \in \overline{CD}$. The same is true for $X_2$. If $i$ is not a multiple of 3 then one of $\overline{v_i X_1}$ or $\overline{v_i X_2}$ is a good dissection. 
\begin{itemize}
\item[$\blacktriangleright$] $i \equiv 1$ mod $3$: 
\begin{itemize}
 \item $X_1 \in \overline{AA_0}$. Then the subpolygon $v_1 v_2 \ldots v_i X_1$ has size $3k'+2$, and the subpolygon $A X_1 v_{i+1} v_{i+2} \ldots v_{n-1}$ has size $3k''+2$ for some $k' +k'' = k$ (recall that $n = 3k+2$).
 \item $X_1 \in \overline{CD}$. Then the polygon $C B A v_1 v_2 \ldots v_i X_1$ has size $3k'+2$, and the subpolygon $X_1 v_{i+1} v_{i+2} \ldots v_{n-3}$ has size $3k''+2$ for some $k' +k'' = k$.
 \end{itemize} 
 \item[$\blacktriangleright$] $i \equiv 2$ mod $3$: 
\begin{itemize}
 \item $X_2 \in \overline{AA_0}$. Then the polygon $v_1 v_2 \ldots v_{i-1} X_2$ has size $3k'+2$, and the subpolygon $A X_2 v_{i} v_{i+1} \ldots v_{n-1}$ has size $3k''+2$ for some $k' +k'' = k$.
 \item $X_2 \in \overline{CD}$. Then the polygon $C B A v_1 v_2 \ldots v_{i-1} X_2$ has size $3k'+2$, and the subpolygon $X_2 v_{i} v_{i+1} \ldots v_{n-3}$ has size $3k''+2$ for some $k' +k'' = k$.
\end{itemize}
\end{itemize}

If $i$ is a multiple of $3$ then we repeat the above procedure and rotate the ray $\overrightarrow{Dv_i}$ towards $v_{i+1}$. Notice that if the ray hits $v_{i+1} \neq D_0$ then $\overline{D v_{i+1}}$ is a good diagonal dissection (which is a contradiction to our main assumption). Thus the ray hits $v_j$ for $i+1 < j \leq v_{n-5}$. If $j$ is not a multiple of $3$ then one of the rays $\overrightarrow{v_{j+1} v_j}$ or $\overrightarrow{v_{j-1} v_j}$ contains a good dissection. Observe that those rays reach $\partial P'$ at $\overline{CD}$, $\overline{AA_0}$ or $\overline{v_i v_{i+1}}$.  Since we perform counting modulo $3$, those edges are considered to be identical in terms of vertices' indices. It means that we do not have to know where exactly $\overrightarrow{v_{j+1} v_j}$ or $\overrightarrow{v_{j-1} v_j}$ reach $\partial P'$ to decide which dissection to apply. That is, if $j \equiv 1$ mod $3$ then we use $\overrightarrow{v_{j+1} v_j}$; if $j \equiv 2$ mod $3$ then we use $\overrightarrow{v_{j-1} v_j}$.

If $j$ is a multiple of $3$ then the procedure is repeated again. Eventually, the ray spinning around $D$ must hit $D_0$. Recall that $D_0 =v_{n-4}$; $n-4 = 3k+2-4 = 3(k-1)+1$, which is not a multiple of $3$. At this point T{\'o}th comes to a contradiction and states that the angle of $P'$ at $D$ cannot be convex (refer to Claim $9$ in~\cite{Toth2000121}). However there is no contradiction. We show that the situation is possible and discuss how to find a good dissection in this case.

Let $v_z$ be the last vertex hit by the ray spinning around $D$ before it hit $D_0$. Notice that $z$ is a multiple of $3$. Two cases are possible:
\begin{enumerate}
\item $v_{z+1} \neq D_0$: notice that $D_0$ can see $v_z$. The size of the subpolygon $v_z v_{z+1} \ldots v_{n-4}$ is $3k' + 2$ for some integer $k'>1$ ($k$ is strictly bigger than $1$ because $v_{z+1} \neq D_0$). Therefore, the diagonal $\overline{v_z D_0}$ is a good diagonal dissection, meaning that this case is not possible.
\item $v_{z+1} = D_0$: in this case $z = n-5$ and the angle of $P'$ at $D_0$ is convex. Refer to Fig.~\ref{fig:TothShortLeaf_1011}. 
\begin{itemize}
\item[$\blacktriangleright$] If $v_z$ can see $A$ then dissect $P'$ along $\overline{v_z A}$. This dissection creates two subpolygons: $Av_zD_0DCB$ of size $6$ and $Av_1 v_2 \ldots v_z$ of size $n-4 = 3(k-1)+1$. Notice that the hexagon $Av_zD_0DCB$ has a non-empty kernel whose intersection with the boundary of $Av_zD_0DCB$ is $\overline{C'B'}$. One $180^\circ$-guard on $\overline{C'B'}$ can monitor $Av_zD_0DCB$. Similarly, for our problem, two distinct towers on $\overline{C'B'}$ can localise an agent in $Av_zD_0DCB$ (notice that $C' \neq B'$ and thus $\overline{C'B'}$ contains at least two distinct points).
\item[$\blacktriangleright$] If $v_z$ cannot see $A$ then consider the ray $\overrightarrow{D_0v_z}$. 
Let $Z$ be a point where $\overrightarrow{D_0v_z}$ reaches $\partial P$ or $\overline{AA'}$ (whichever happens first).
\begin{itemize}
\item If $Z \in \overline{AA'}$ (refer to Fig.~\ref{fig:TothShortLeaf_10}) then dissect $P'$ along the two line segments $\overline{AZ} \cup \overline{Zv_z}$. $P'$ falls into two subpolygons: $AZv_zD_0DCB$ of size $7$ (which is technically a hexagon since $Z$, $v_z$ and $D_0$ are collinear) and $ZAv_1 v_2 \ldots v_z$ of size $n-3 = 3(k-1)+2$. Similarly to the hexagon from the previous case, the heptagon $AZv_zD_0DCB$ can be guarded by one $180^\circ$-guard on $\overline{C'B'}$ and our agent can be localised by a pair of distinct towers positioned on $\overline{C'B'}$.
\item If $Z \notin \overline{AA'}$ (refer to Fig.~\ref{fig:TothShortLeaf_11}) then there must be a vertex $v_x$ to the left of $\overrightarrow{AA'}$ visible to $D_0$ and to $D$. Recall that for every vertex $v_x$ visible from $D$, the index $x$ is a multiple of $3$. 
Thus $\overline{v_xD_0}$ is a good diagonal dissection. It means that this case is not possible.
\end{itemize}
\end{itemize}
\end{enumerate}

\begin{figure}
\centering
\subfigure[]{%
		\label{fig:TothShortLeaf_10}%
		\includegraphics[width=0.4\textwidth]{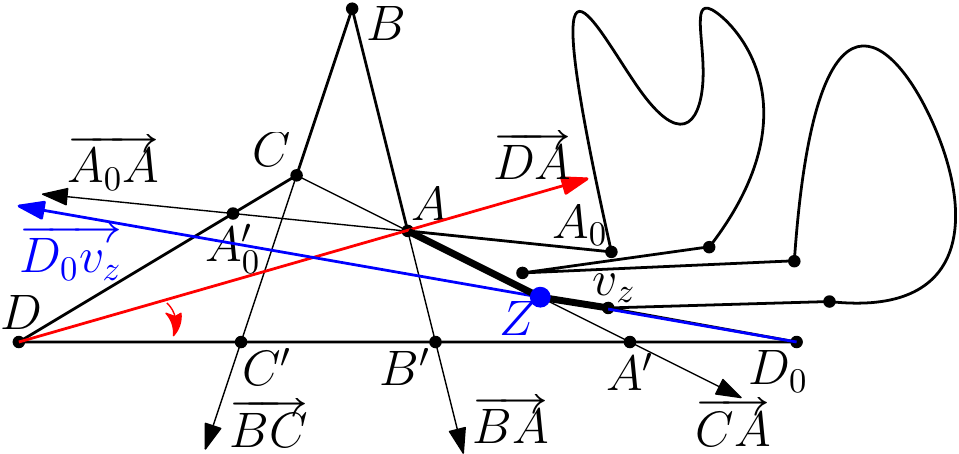}}
\hspace{0.1\textwidth}
\subfigure[]{%
		\label{fig:TothShortLeaf_11}%
		\includegraphics[width=0.4\textwidth]{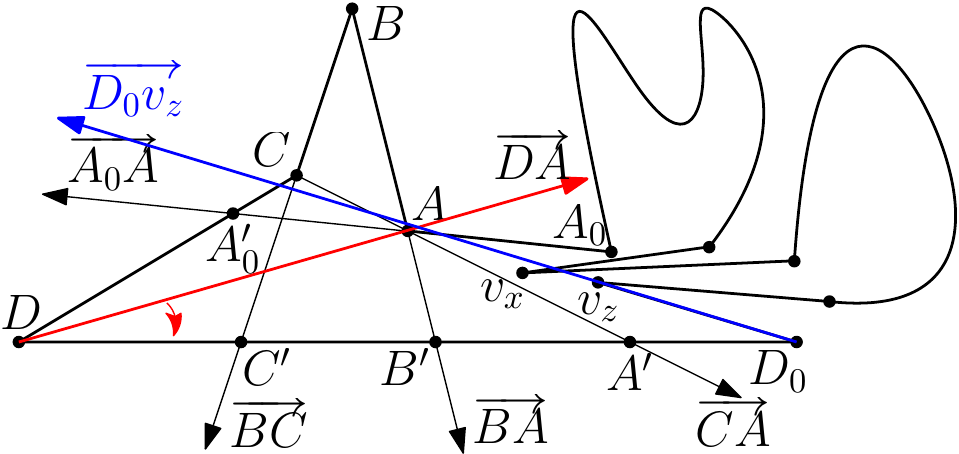}}
\caption{$\triangle ABC$ corresponds to a long leaf in $G(\mathcal{T})$; $\angle D_0 D C < \pi$; $A_0' \in \overline{CD}$. \textbf{(a)}Two line segments $\overline{AZ} \cup \overline{Zv_z}$ is a good dissection. \textbf{(b)} Impossible case, because $\overline{v_xD_0}$ is a good diagonal dissection.}
\label{fig:TothShortLeaf_1011}
\end{figure}

This completes the proof of T{\'o}th's Lemma~\ref{TothLem2}.

\subsection{Partition algorithm}
\label{subsec:PartitionAlg}
\pdfbookmark[2]{Partition algorithm}{subsec:PartitionAlg}

We are given a simple polygon $P$ in general position of size $n = 3k +q$ where $k$ is a positive integer and $q \in \{0, 1, 2\}$. If $P$ has at most $k$ reflex angles then $P$ can be partitioned simply by bisecting $k-1$ of the reflex angles. This creates $k$ star-shaped subpolygons of $P$ that can be watched by $2k$ towers. If the number of reflex angles of $P$ is bigger than $k$ then we look for a good diagonal dissection. In Section~\ref{subsec:PointKernel} we modified T{\'o}th's partition~\cite{Toth2000121} to deal with subpolygons of $P$ (polygons*) that are not in general position. The important difference to notice is that we avoided dissecting along diagonals of $P'$ that contain vertices of $P$ in their interior. If the dissection is unavoidable we showed how to position towers and in the worst case - repartition subpolygons of $P'$.

If $P'$ has no good diagonal dissection then its size is $n=3k+2$. In this case we look for a good dissection via the \emph{short} or \emph{long leaf} approach discussed in Section~\ref{subsubsec:Lemma2}. If no good cut is found then by T{\'o}th's Lemma~\ref{TothLem2} every leaf in $G(\mathcal{T})$ is associated to two convex angles of $P'$. It follows by T{\'o}th's Lemma~\ref{TothLem3} that $P'$ can be monitored by $\left \lfloor \frac{n}{3} \right \rfloor$ $180^\circ$-guards. Refer to Section~\ref{subsubsec:Lemma3} on how to adapt T{\'o}th's Lemma~\ref{TothLem3} for tower positioning. We refer to T{\'o}th's Lemma~\ref{TothLem1} to show the coherence of the algorithm. T{\'o}th's Lemma~\ref{TothLem1} states that if $P'$ has $n = 3k+2$ vertices and has no good diagonal or other dissection then $P'$ has at most $k$ reflex angles and thus $P'$ would be treated during the first step of the algorithm. 

The obtained partition together with the locations of $180^\circ$-guards is reused for tower positioning. Every $180^\circ$-guard that guards subpolygon $P'$ is positioned on the boundary of $P'$ (either on an edge or a convex vertex of $P'$) and oriented in such a way that $P'$ completely belongs to the half-plane $H_l$ monitored by the guard. In our problem every $180^\circ$-guard is replaced by a pair of towers $t_1$ and $t_2$ on the same edge of $P'$ in $\partial P' \cap kernel(P')$ and close to the $180^\circ$-guard. The orientation of $180^\circ$-guard is embedded into the tower coordinates via the \textbf{parity trick}. Specifically, given the partition of $P$ we can calculate a line segment $\overline{ab}$ suitable for tower positioning for each polygon $P'$ of the partition. Towers can be anywhere on $\overline{ab}$ as long as the distance between them is a rational number. Notice, that the coordinates of the endpoints of $\overline{ab}$ can be irrational numbers because $a$ and $b$ can be vertices of $P$ or $P'$ or intersection points between lines that support edges of $P$ or $P'$. 
However, we will show that it is always possible to to find the desired positions for the two towers: We place the tower $t_1$ at the point $a$. The tower $t_2$ will be placed at the point $c = a + \lambda (b-a)$ for an appropriate choice of $\lambda$. Let $s \geq 1$ be the smallest integer such that $1/3^s \leq d(a,b)$. If we take $\lambda = \frac{1}{3^{s+1} d(a,b)}$, then $d(a,c) = d(t_1,t_2) = \frac{1}{3^{s+1}}$, which is a reduced rational number whose numerator is odd. On the other hand, if we take $\lambda = \frac{2}{3^{s+1} d(a,b)}$, then $d(a,c) = d(t_1,t_2) = \frac{2}{3^{s+1}}$, which is a reduced rational number whose numerator is even. 
If we want the pair of towers $t_1$ and $t_2$ to be responsible for the half-plane $L(t_1, t_2)^+$ (respectively $L(t_1, t_2)^-$) we position the towers at $\frac{2}{3^{t+1}}$ (respectively $\frac{1}{3^{t+1}}$) distance from each other. 

Notice that $L(t_1, t_2)$ is not always parallel to the line that supports $H_l$. If the $180^\circ$-guard is positioned at a convex vertex $v$ of $P'$ then only one tower can be positioned at $v$. Another tower is placed on the edge adjacent to $v$ in $\partial P' \cap kernel(P')$.
If $kernel(P')$ is a single point then we position our towers outside of $P'$ and close to $kernel(P')$, such that $L(t_1, t_2)$ is parallel to the line that supports $H_l$. If $L(t_1, t_2)^+$ (respectively $L(t_1, t_2)^-$) contains $H_l$ then we position the towers at a distance, which is a rational number whose numerator is even (respectively odd).

Algorithm~\ref{alg_Partition} partitions $P$ and positions at most $\left \lfloor \frac{2n}{3} \right \rfloor$ towers that can localize an agent anywhere in $P$. The localization algorithm (Algorithm~\ref{alg_Locate}) can be found in Section~\ref{sec:localization}.

\begin{algorithm}
\caption{Polygon Partition; Tower Positioning}\label{alg_Partition}
\KwIn{$P'$ of size $n = 3k +q$, for positive integer $k$ and $q \in \{0, 1, 2\}$}
\KwOut{Set of at most $\left \lfloor \frac{2n}{3} \right \rfloor$ towers in $P'$}
\BlankLine
If $P'$ has at most $k$ reflex angles then position $2k$ towers by bisecting reflex angles. Halt\;
Simplify $P'$\;
\eIf{there exists a good diagonal dissection that contains at most one vertex of $P$ in its interior}{apply it and run this algorithm on $P_1$ and $P_2$}
(\tcp*[f]{$n = 3k +2$})
{
\eIf{there exist a good dissection via a continuation of an edge}{apply it; run this algorithm on $P_1$ and $P_2$}{
\eIf{there exists a good diagonal dissection}
{this dissection contains two vertices of $P$ in its interior\;
apply it and run this algorithm on $P_2$\;
\eIf{pentagon with a pair of vertices of $P$ in the interior of the same edge is created}{Repartition $P_2$ as described in Section~\ref{subsec:PointKernel}}{run this algorithm on $P_1$}
}
(\tcp*[f]{$n = 3k +2$ and $P'$ has no good diagonal dissection})
{
{\If{$P'$ has a good dissection via \emph{short} or \emph{long leaf} approach (refer to~\cite{Toth2000121} and Section~\ref{subsubsec:Lemma2})}{use it; repeat algorithm on $P_1$ and $P_2$}
  }
  }
}
}
\end{algorithm}

The running time of Algorithm~\ref{alg_Partition} is $O(n^3)$ because of the cases where repartitioning of already partitioned subpolygons is required.

In Section~\ref{sec:partition} we showed how to use the polygon partition method introduced by T{\'o}th~\cite{Toth2000121} for wider range of polygons by lifting his assumption that the partition method creates subpolygons whose vertices are in general position. We reproved T{\'o}th's results and showed how to use his partition method for localization problem. We showed how to compute a set $T$ of size at most $\lfloor 2n/3\rfloor$ that can localize an agent anywhere in $P$. The results of Section~\ref{sec:partition} are summarized in the following theorem.

\begin{theorem}
\label{thm:main_result}
Given a simple polygon $P$ in general position having a total of $n$ vertices. Algorithm~\ref{alg_Partition} computes a set $T$ of broadcast towers of cardinality at most $\lfloor 2n/3\rfloor$, such that any point interior to $P$ can localize itself.
\end{theorem}

\newpage
\subsection{Counterexample to T{\'o}th's conjecture}
\label{sec:counterexample}
\pdfbookmark[2]{Counterexample to T{\'o}th's conjecture}{sec:counterexample}

We refute the conjecture given by T{\'o}th in~\cite{Toth2000121}.

\vspace{10pt}
\textbf{Conjecture:} Any simple polygon of $n$ sides can be guarded by $\lfloor n/3\rfloor$ $180^\circ$-guards that are located exclusively on the boundary of the polygon.

Figure~\ref{fig:counterexample} shows a simple polygon $P$ with $n=8$ that can be guarded by $2$ general $180^\circ$-guards but requires at least $3$ $180^\circ$-guards that must reside on the boundary of $P$. Observe that $P$ is not a star-shaped polygon.
Figure~\ref{fig:counterexample11} shows a possible partition of $P$ into two star-shaped polygons: $v_1 v_2 v_3 v_4 v_5$ and $v_1 v_5 v_6 v_7 v_8$. 
\begin{figure}
\centering
\subfigure[]{%
		\label{fig:counterexample1}%
		\includegraphics[width=0.25\textwidth]{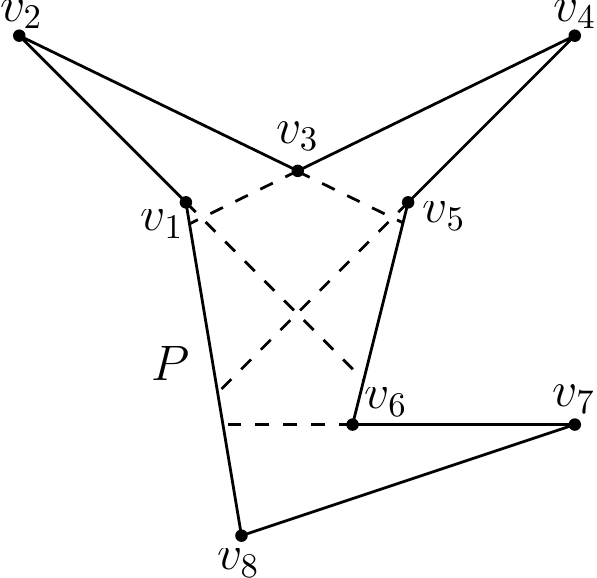}}
\hspace{0.1\textwidth}
\subfigure[]{%
		\label{fig:counterexample11}%
		\includegraphics[width=0.25\textwidth]{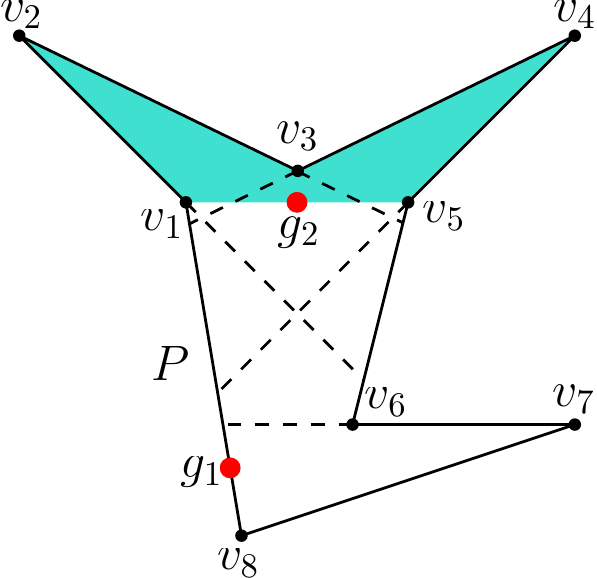}}
\hspace{0.1\textwidth}
\subfigure[]{%
		\label{fig:counterexample12}%
		\includegraphics[width=0.25\textwidth]{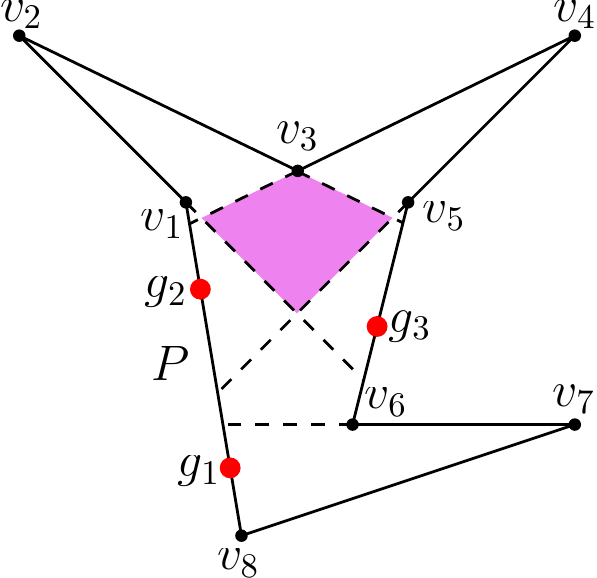}}
\caption{Counterexample on $n=8$ vertices. Guards are highlighted in red. \textbf{(b)} $P$ can be guarded by $2$ general $180^\circ$-guards\textbf{(c)} $P$ requires at least $3$ $180^\circ$-guards that must reside on the boundary of $P$.}
\label{fig:counterexample}
\end{figure}
The polygon $v_1 v_2 v_3 v_4 v_5$ can be guarded by the $180^\circ$-guard $g_2$ located on $\overline{v_1 v_5}$ and  oriented upwards (i.e. $g_2$ observes $L(v_1, v_5)^+$). The second $180^\circ$-guard $g_1$ is located on $\overline{v_1 v_8}$ in $kernel(v_1 v_5 v_6 v_7 v_8)$. It is  oriented to the right of $L(v_1, v_8)$ (i.e. $g_1$ observes $L(v_1, v_8)^-$) and thus guards $v_1 v_5 v_6 v_7 v_8$. Consider Figure~\ref{fig:counterexample12}. The visibility region, from where the complete interior of $v_1 v_2 v_3 v_4 v_5$ can be seen, is highlighted in magenta. We want to assign a single $180^\circ$-guard that can see both vertices $v_2$ and $v_4$ and be located on $\partial P$. Notice that the intersection of this visibility region with $\partial P$ contains a single point $v_3$. However, the angle of $P$ at $v_3$ is reflex and the guards have a restricted $180^\circ$ field of vision. Thus it is impossible to guard $v_1 v_2 v_3 v_4 v_5$ with a single $180^\circ$-guard located on $\partial P$. Notice that the visibility region of the vertex $v_7$ does not intersect with the visibility region of $v_2$ and the visibility region of $v_4$. Thus it requires an additional guard. It follows that $P$ requires at least $3$ $180^\circ$-guards located on $\partial P$. Notice that $\lfloor n/3\rfloor = \lfloor 8/3\rfloor = 2$. This is a contradiction to the above conjecture.

In general, consider the polygon shown in Fig.~\ref{fig:counterexample2}. It has $n = 5s + 2$ vertices, where $s$ is the number of \emph{double-spikes}. Each spike requires its own guard on the boundary of $P$ (two guards per double-spike), resulting in $2s$ boundary guards in total. This number is strictly bigger than $\lfloor n/3\rfloor = \lfloor 5s + 2/3\rfloor$ for $s \geq 3$.

\begin{figure}[h]
\centering
\includegraphics[width=0.8\textwidth]{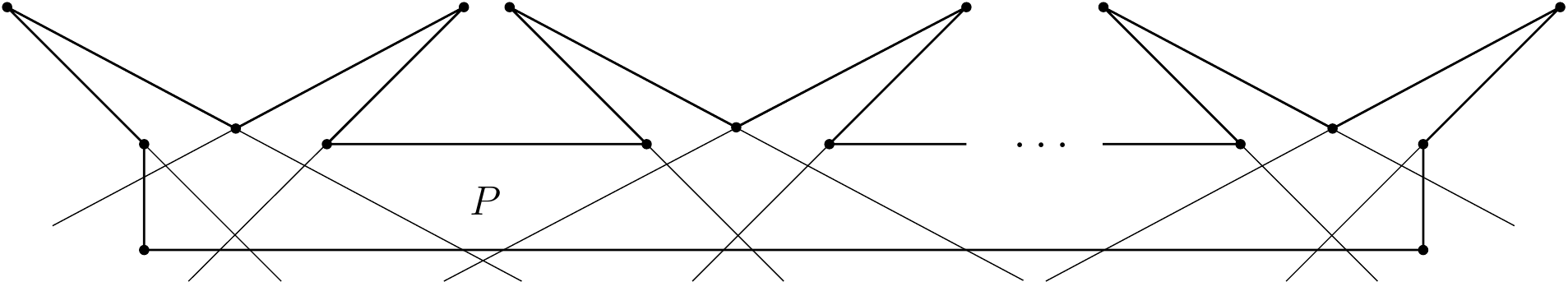}
\caption{Counterexample to T{\'o}th's conjecture. The polygon $P$ is in general position.}
\label{fig:counterexample2}
\end{figure}

\section{Localization Algorithm}
\label{sec:localization}
\pdfbookmark[1]{Localization Algorithm}{sec:localization}

In Section~\ref{sec:partition} we showed how to position at most $\left \lfloor \frac{2n}{3} \right \rfloor$ towers in a given polygon. We used a modification of T{\'o}th's partition method that dissects a polygon into at most $\lfloor n/3\rfloor$ star-shaped polygons each of which can be monitored by a pair of towers. In this section we show how we can localize an agent $p$ in the polygon. Our localization algorithm receives as input only the coordinates of the towers that can see $p$ together with their distances to $p$. In this sense, our algorithm uses the classical trilateration input. In addition, our algorithm knows that the parity trick was used to position the towers. Based on this information alone, and without any additional information about $P$, the agent can be localized. When only a pair of towers $t_1$ and $t_2$ can see the point $p \in P$ then the coordinates of the towers together with the distances $d(t_1,p)$ and $d(t_2,p)$ provide sufficient information to narrow the possible locations of $p$ down to two. Those two locations are reflections of each other over the line through $t_1$ and $t_2$. In this situation our localization algorithm uses the parity trick. It calculates the distance between the two towers and judging by the parity of the numerator of this rational number decides which of the two possible locations is the correct one. Refer to Algorithm~\ref{alg_Locate}.

\begin{algorithm}[h]
\caption{Compute the coordinates of point $p$.}\label{alg_Locate}
\KwIn{\begin{list}{}{}
\item $t_1$, \ldots, $t_k$ -- coordinates of the towers that see $p$.
\item $d_1$, \ldots, $d_{\ell}$ -- distances between the corresponding towers and $p$.
\end{list}}
\KwOut{coordinates of $p$.}
\BlankLine
\eIf{$\ell \geq 3$}{
	$p = C(t_1, d_1) \cap C(t_2, d_2) \cap C(t_3, d_3)$\;
}(\tcp*[f]{$\ell = 2$})
{
	\eIf{the numerator of $d(t_1,t_2)$ is even}{
	$p = C(t_1, d_1) \cap C(t_2, d_2) \cap L(t_1, t_2)^+$\;
	}
	{
	$p = C(t_1, d_1) \cap C(t_2, d_2) \cap L(t_1, t_2)^-$\;
	}
}
Return $p$\;
\end{algorithm}

\newpage
\section{Concluding Remarks}
\label{sec:conclusion}
\pdfbookmark[1]{Concluding Remarks}{sec:conclusion}

We presented a tower-positioning algorithm that computes a set of size at most $\lfloor 2n/3\rfloor$ towers, which improves the previous upper bound of $\lfloor 8n/9\rfloor$~\cite{DBLP:conf/cccg/DippelS15}. 
We strengthened the work~\cite{Toth2000121} by lifting the assumption that the polygon partition produces polygons whose vertices are in general position. We reproved T{\'o}th's result. We found and fixed mistakes in claims $2$ and $7$ in~\cite{Toth2000121}.

We believe it is possible to avoid the repartition step (described in Section~\ref{subsec:PointKernel}) and as a consequence bring the running time of Algorithm~\ref{alg_Partition} to $O(n^2)$ instead of $O(n^3)$.

As a topic for future research we would like to show that determining an optimal number of towers for polygon trilateration is NP-hard.

%
%

\bibliography{Trilateration}

\end{document}